\documentclass[journal,letterpaper,onecolumn,11pt]{IEEEtran}
\usepackage{hyperref}
\usepackage{tikz}
\usetikzlibrary{positioning, arrows.meta, fit, backgrounds, calc}
\usetikzlibrary{arrows.meta}
\usepackage{amsmath}
\hypersetup{
 colorlinks,
 linkcolor={blue!100!black},
 citecolor={blue!100!black},
 urlcolor={blue!80!black}
}
\usepackage{accents}
\newlength{\dhatheight}

\def\BibTeX{{\rm B\kern-.05em{\sc i\kern-.025em b}\kern-.08em
    T\kern-.1667em\lower.7ex\hbox{E}\kern-.125emX}}
\usepackage{stmaryrd} 

\newcommand{\bN}{\mathbb{N}}

\newcommand{\bZ}{\mathbb{Z}}
\usepackage{booktabs} 
\newcommand{\cA}{\mathcal{A}}
\newcommand{\cB}{\mathcal{B}}
\newcommand{\cC}{\mathcal{C}}

\newcommand{\cE}{\mathcal{E}}
\newcommand{\cF}{\mathcal{F}}
\newcommand{\cG}{\mathcal{G}}

\newcommand{\cM}{\mathcal{M}}

\newcommand{\cP}{\mathcal{P}}

\newcommand{\cR}{\mathcal{R}}
\newcommand{\cS}{\mathcal{S}}
\newcommand{\cT}{\mathcal{T}}
\newcommand{\cU}{\mathcal{U}}

\newcommand{\boldc}{\mathbf{c}}

\newcommand{\boldf}{\mathbf{f}}

\newcommand{\boldm}{\mathbf{m}}

\newcommand{\boldp}{\mathbf{p}}
\newcommand{\boldq}{\mathbf{q}}

\newcommand{\bolds}{\mathbf{s}}

\newcommand{\boldu}{\mathbf{u}}
\newcommand{\boldv}{\mathbf{v}}
\newcommand{\boldw}{\mathbf{w}}
\newcommand{\boldx}{\mathbf{x}}
\newcommand{\boldy}{\mathbf{y}}
\newcommand{\boldz}{\mathbf{z}}

\newcommand{\boldH}{\mathbf{H}}

\newcommand{\rank}{\operatorname{rank}}

\usepackage{tikz}
\newif\ifFULL
\FULLfalse
\usepackage[utf8]{inputenc} 
\usepackage[T1]{fontenc}
\usepackage{url}
\usepackage{bbm}
\usepackage{ifthen}
\usepackage[cmex10]{mathtools}

\usepackage{verbatim}
\usepackage{amsthm,amssymb}
\usepackage{xcolor}
\usepackage{graphicx}
\usepackage{amssymb}
\usepackage{epstopdf}

\newtheorem{theorem}{Theorem}
\newtheorem{observation}{Observation}
\newtheorem{remark}{Remark}

\newtheorem{proposition}{Proposition}
\newtheorem{example}{Example}
\newtheorem{lemma}{Lemma}
\newtheorem{definition}{Definition}

\newtheorem{corollary}{Corollary}
\usepackage[style=ieee,citestyle=numeric-comp,sorting=nty]{biblatex}
\usepackage{algorithm}
\usepackage{algorithmic}
\begin{document}
\title{Capacity Achieving Torn Paper Codes} 


\author{%
  \IEEEauthorblockN{\textbf{Junsheng Liu} and \textbf{Netanel Raviv}}  \IEEEauthorblockA{\\~\\Department of Computer Science and Engineering\\
                    Washington University in St Louis, St Louis, MO, USA\\
                    \texttt{junsheng,netanel.raviv@wustl.edu}}
}


\maketitle


\begin{abstract}
In the torn paper channel, a codeword is cut at random locations, and the resulting error-free fragments are delivered to the decoder as an unordered multiset. 
Although the capacity of this channel can be achieved using random codebooks, decoding such codes generally requires exponential time, whereas existing structured schemes with efficient decoding and provable vanishing error remain strictly below the capacity. 
The interleaved-pilot construction of Shomorony and Vahid embeds a De Bruijn sequence among the symbols of a shifted erasure code and aligns only fragments that are sufficiently long for the pilot to be distinguished from the information symbols through global statistical uniqueness. 
A subsequent local-alignment scheme by Liu and Raviv employs run-length-limited constraints and exclusive all-zero markers to identify pilot positions from local structure, substantially reducing the minimum fragment length that can be aligned. 
Although the second method improves the achievable rate, it still falls short of the torn-paper-channel capacity.

We further improve the method of Liu and Raviv by replacing its fixed pilot sequence with a multilevel \emph{successive local alignment and pilot-recycling} procedure. 
In Liu and Raviv, the pilot sequence is chosen once and must simultaneously balance the length of the pilot sequence against the ability to align short fragments, which prevents the achievable rate from reaching capacity. 
Our construction removes this limitation by reusing pilot sequences across successive decoding levels. 
A pilot sequence with an independent random linear code is first used to align and decode the longest fragments. 
The information recovered at this stage is then recycled as a larger pilot sequence for the next stage.
This process is repeated over multiple levels, so that progressively shorter fragments are recovered while the length of the pilot sequence remains small, and thus reaches a higher achievable rate.
Our construction depends on choosing a series of random linear codes, and we show that for any~$\varepsilon>0$, a choice of such codes which attains rate at least~$\varepsilon$ below the capacity, with high probability as the block length goes to infinity.
Therefore, our construction achieves the capacity of the torn-paper channel.  

\end{abstract}

\section{Introduction}
{\let\thefootnote\relax\footnote{
This work was supported in part by NSF grant CNS-2223032.}}

Consider a simple reconstruction problem. An information sequence is written on a long strip of paper, the strip is torn at random locations, and the resulting pieces are shuffled before reaching the receiver. Although every bit on each individual fragment remains intact and readable, all information about the original order of the fragments is lost. 
Recovering the message therefore requires determining where every useful fragment originated in the sequence. This is the central difficulty of \emph{torn paper coding}.

The problem is motivated by storage and identification systems in which physical fragmentation can destroy positional context without necessarily corrupting the symbols themselves. DNA storage, for example, has emerged as a promising medium for large-scale archival data~\cite{douglas2012logic,goldman2013towards,blawat2016forward,antkowiak2020low}, while DNA molecules are inherently vulnerable to strand breaks. A similar issue arises when identifying 3D-printed objects from information embedded in their physical structure~\cite{wang2024secureinformationembeddingextraction}: once an object is broken, the embedded sequence is distributed among separate pieces whose original locations are no longer apparent. These applications lead naturally to a coding-theoretic abstraction in which a codeword is cut randomly into error-free fragments and the decoder observes only their unordered multiset.

In the torn paper channel, introduced by Shomorony and Vahid~\cite{shomorony2021torn}, a break occurs independently between consecutive codeword symbols with probability $p_n$. 
It was shown in~\cite{shomorony2021torn} that the channel capacity is $e^{-\alpha}$, where $\alpha=\lim_{n\to\infty}p_n\log n$, and that purely random code achieves this capacity. 
Random coding, however, generally requires exponential decoding complexity and is therefore impractical for efficient alignment and reconstruction.

Several works have studied coding specifically for the torn paper channel. 
Shomorony and Vahid~\cite{shomorony2021torn} developed a structured interleaved-pilot construction with a rigorous proof of vanishing decoding error and a provable achievable rate. 
Building on their framework, Liu et al.~\cite{liu2026improved} introduced a local alignment method that reduces the minimum usable-fragment length and consequently achieves a strictly higher provable rate.
Nevertheless, both structured constructions remain strictly below the torn paper channel capacity.
Some experimental results using nested constructions have been proposed in~\cite{nassirpour2023dna,jiao2025efficient} and are shown to achieve favorable empirical rates. 
However, these works do not address large codeword length $n$ and large break probability $\alpha$, and their performance evaluations are primarily experimental, without rigorous analytical proofs of the achievable rates or decoding reliability.

Beyond the original model, a number of related fragmentation settings have also been considered. These include the torn paper channel with lost pieces, in which fragments are deleted according to length-dependent probabilities~\cite{ravi2024recovering}; break-resilient coding for forensic 3D printing~\cite{wang2024breakresilientcodesforensic3d}; reconstruction from a single sufficiently long fragment~\cite{liu2025single}; adversarial torn paper coding~\cite{barlev}; and the sliced channel~\cite{sima2023error}. In this paper, we focus on the original  torn paper channel and aim to close the gap between its known capacity and the achievable rates of coding constructions with provable reliability.

The first provable coding scheme was proposed by Shomorony and Vahid~\cite{shomorony2021torn}. 
Their construction regularly inserts known pilot symbols into the codeword. 
When a fragment is received, the decoder first searches for these pilot symbols and uses them to determine where the fragment came from in the original codeword. 
This works reliably only when the fragment is long enough to contain a sufficiently long pilot subsequence. 
Shorter fragments cannot be positioned with confidence and are therefore discarded. 
As a result, fragments shorter than the required alignment length are discarded, and the loss of these fragments prevents the construction from achieving capacity.
We improved the achievable rate in our previous work~\cite{liu2026improved} by allowing the decoder to use much shorter fragments. 
The main idea was to place special all-zero markers at known locations in the pilot sequence and to encode the information symbols using a Run-Length-Limited code so that the same marker could not appear inside the data. 
Therefore, when a fragment contains one of these markers, the decoder can quickly identify the pilot symbols from the fragment itself. 
It can then use the nearby pilot pattern to determine where the fragment belongs in the original codeword. 
Unlike the earlier method, this local alignment approach does not require a long fragment to distinguish the pilot from the data, so fewer received fragments are discarded.
This leads to a higher provable rate while still guaranteeing that the decoding error probability vanishes. 
However, the construction uses only one fixed pilot throughout the codeword. 
Using more pilot symbols helps recover shorter fragments but leaves less space for information, whereas using fewer pilot symbols reserves more space for information but causes more short fragments to be discarded. 
Due to this remaining tradeoff, the method improves the rate but falls short of achieving the torn paper channel capacity.

The present work improves the previous local alignment method through \emph{successive local alignment with pilot recycling}. 
Instead of using one fixed pilot pattern to align all fragments, our construction works in several decoding levels. 
A modified De Bruijn pilot sequence is inserted to begin the decoding process. 
At first, this pilot sequence is used to align the longest received fragments, which provide enough information to decode the first information level.
The main improvement is that the decoded information is not used only as message data. 
Once a level has been decoded, its known symbols are reused as additional pilot symbols for the next level. 
The decoder therefore has many more known positions than it had at the beginning. 
Using these newly available positions, it can align fragments that are too short to be aligned by the original pilot sequence. 
After another information level is decoded, its symbols are also recycled as pilots. 
Repeating this process allows the decoder to use progressively shorter fragments without inserting a large number of fixed pilot symbols throughout the codeword.
In this way, the construction avoids the main tradeoff of the previous method: it uses only the initial pilot sequence while still recovering many of the short fragments that would otherwise be discarded.
The advantage of the multilevel construction is that different sets of fragments are recovered at different stages. The first stage uses the longest fragments, and each later stage adds progressively shorter ones. 
By increasing the number of levels and making the difference of fragment length threshold in adjacent levels scales sufficiently small, the total recovered information approaches the capacity of the channel. To be more specific, the construction can be chosen to achieve a rate  for any desired gap from capacity while maintaining vanishing average decoding error. 

The remainder of the paper is organized as follows. Section~\ref{Section:problem definition} formally defines the torn paper channel, recalls its capacity, reviews the interleaved pilot construction of~\cite{shomorony2021torn} and the local alignment construction of~\cite{liu2026improved}.
Section~\ref{section:successive alignment} introduces the successive local alignment and pilot-recycling framework and defines the multilevel residue structure used throughout the construction.
Section~\ref{Section:proposed-work encoding} introduces the proposed coding method and explains how the codeword is divided into several information levels together with an initial pilot sequence. Section~\ref{Section:levelized-decoding} describes the decoding process, beginning with the longest fragments and then reusing the recovered information as new pilot symbols to align progressively shorter fragments. 
Section~\ref{Section:rate-and-reliability} proves that the proposed method is reliable, shows that the loss caused by ambiguous fragments becomes negligible, and establishes that the achievable rate approaches the torn-paper-channel capacity. \
\section{Problem Definition and Previous Results}\label{Section:problem definition}

\subsection{The torn paper channel}\label{section:torn-paper-channel}
We consider the problem of encoding 
a message in $\{0,1\}^k$ into a codeword~$\boldx=(\boldx[1],\ldots,\boldx[n])\in\{0,1\}^n$. 
During storage or transmission, the codeword is subject to random physical breaks between consecutive bits. 
For each index $i\in \{1,2,\dots,n-1\}$, a break occurs between $\boldx[i]$ and $\boldx[i+1]$ with probability $p_n$ and no break with probability $1-p_n$, where~$p_n$ is a channel parameter which may depend on~$n$.
The fragmentation process does not delete or alter any bits; it only introduces random cuts between consecutive positions. 
More specifically, let $N_1, N_2, \ldots$ be i.i.d. $\operatorname{Geometric}(p_n)$ random variables representing the lengths of the broken pieces, and let $K$ be the smallest index such that $$\sum_{i=1}^{K} N_i \ge n.$$
The channel tears the string $\boldx$ into fragments $\vec{X}_1, \vec{X}_2, \ldots, \vec{X}_K$, where 
$$\vec{X}_i \triangleq \left( \boldx[1+\textstyle\sum_{j=1}^{i-1} N_j], \dots, \boldx[\textstyle\sum_{j=1}^i N_j] \right)$$
for~$1\le i<K$, and the final fragment is bounded by the total length $n$:
$$\vec{X}_K \triangleq \left( \boldx[1+\textstyle\sum_{j=1}^{K-1} N_j], \dots, \boldx[n] \right).$$

The decoder receives the unordered multiset $\{\{\vec{X}_1, \vec{X}_2, \dots, \vec{X}_K\}\}$ of all resulting fragments, and the goal is to design a family of binary codewords $\mathcal{C} \subseteq \{0,1\}^n$ such that the average probability of decoding error—defined as
\[
P_e
=
\frac{1}{|\mathcal{C}|}
\sum_{\boldc\in\mathcal{C}}
\Pr\!\left(
\psi\!\left(\{\{\vec{X}_j\}\}_{j=1}^K\right)\neq \boldc
\right),
\]
for a decoding function \(\psi\) that maps the multiset of fragments back to the original codeword—vanishes as \(n\to\infty\).
Clearly, for any practical use the function~$\psi$ must be efficiently computable. 
The channel is illustrated in Figure~\ref{Figure: channel}, where raw information is encoded into a string, the string is torn into pieces, and the decoder must decode the correct information based on the unordered set of substrings.

We assume that the received fragments retain their orientation, i.e., when receiving a fragment~$(\boldx[i],\ldots,\boldx[i+\ell])$, the decoder does not confuse this fragment with its reverse $(\boldx[i+\ell],\ldots,\boldx[i])$.
This assumption reflects DNA storage settings, where the underlying biochemical directionality of the DNA strand dictates the correct reading direction~\cite{wiki:Directionality_molecular_biology} and is similarly mirrored in 3D-printing forensics where asymmetric inner codes and orientation-specific markers are embedded into the model’s geometry to break rotational symmetry and ensure unique readability~\cite{wang2024secureinformationembeddingextraction}.
Thus, the orientation of each fragment is assumed to be known.

\begin{figure}[ht]
    \centering
    \includegraphics[width=0.6\linewidth]{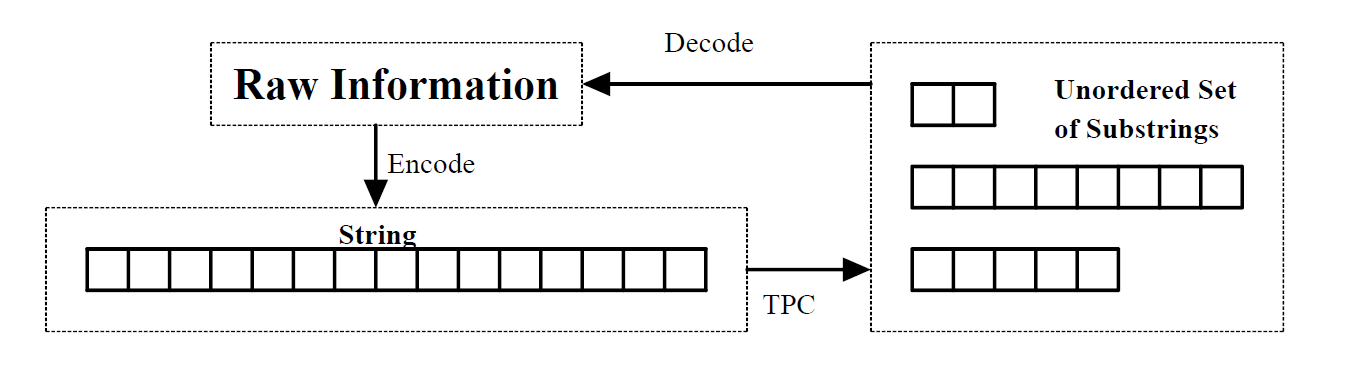}
    \caption{A sketch of torn paper coding.}
    \label{Figure: channel}
\end{figure}

The capacity of the associated channel (called the \textit{torn paper channel}) has been studied in~\cite{shomorony2021torn}.
They defined $\alpha=\lim_{n\to\infty}p_n\log n$ and proved that the capacity of this channel is $e^{-\alpha}$.
In the case that $p_n=o(\frac{1}{\log n})$, we have that $e^{-\alpha}$ approaches~$1$ when $n\xrightarrow[]{}\infty$, and similarly, if $p_n=\omega(\frac{1}{\log n})$ then $e^{-\alpha}$ approaches~$0$ when $n\xrightarrow[]{}\infty$.
Thus, the remaining regime of interest is $p_n=\Theta(\frac{1}{\log n})$, i.e., where $\alpha$ is constant.
\subsection{The interleaved-pilot method of Shomorony and Vahid}\label{section:interleaved pilot scheme}
The first structured construction with a provable vanishing error probability is the interleaved-pilot scheme of~\cite{shomorony2021torn}. 
We begin with the definition of a De Bruijn sequence. 
\begin{definition}\label{Definition:De Bruijn sequence}
A \emph{De Bruijn sequence} of order $\log n$ over $\mathbb{F}_2$ is a binary sequence of length $n$ in which every possible binary string of length $\log n$ appears at most once as a (possibly cyclic) contiguous substring.
\end{definition}

For a fixed positive integer $m$, consider a De Bruijn sequence $\mathbf{p}$ of length $n/m$, which we refer to as the \emph{pilot} sequence. 
According to Definition~\ref{Definition:De Bruijn sequence}, every sequence of length $\log(n/m)$ appears in $\mathbf{p}$ at most once.
To construct the codebook, the authors of~\cite{shomorony2021torn} interleave the pilot sequence with codewords of an erasure code. 
Specifically, let $\cC_{\mathrm{er}}$ be an erasure code of length $n/m$ and rate $R_{\mathrm{er}}$. They form a \emph{shifted} version of this code by drawing an i.i.d.\ $\operatorname{Bernoulli}(1/2)$ random vector in $\{0,1\}^{n/m}$ and XORing it with each codeword of $\cC_{\mathrm{er}}$, thereby obtaining the code $\Tilde{\cC}_{\mathrm{er}}$.
A key property required for successful global alignment is that the pilot sequence substrings should not coincide with any substring of a randomly shifted erasure codeword as follows.
\begin{observation}\label{Lemma: probability of match in ilan's method}
    For any $\mathbf{s} \in \Tilde{\cC}_{\mathrm{er}}$, the probability that $\mathbf{s}$ shares a length-$k$ substring with the pilot sequence is
\begin{align*}
    \Pr\Big(&\mathbf{p}[i:i+k-1]=\mathbf{s}[j:j+k-1] \text{ for any distinct }i,j\in[\tfrac{n}{m}-k]
    \Big) \leq (\tfrac{n}{m})^22^{-k},
\end{align*} 
where the probability is taken over the choices of the random vector used to construct~$\tilde{\cC}_\text{er}$ from~$\cC_\text{er}$, and where $\mathbf{p}[i:i+k-1]$ refers to the $i$-th through $(i+k-1)$-th bits of~$\mathbf{p}$.
\end{observation}

Observation~\ref{Lemma: probability of match in ilan's method} is proved via a straightforward application of the union bound. 
Therefore, if $k=(2+\delta)\log n$ for some $\delta>0$, this probability goes to $0$ when $n\xrightarrow[]{}\infty$. 
This means that for any $\epsilon>0$ there exists a large enough~$n$ so that it is possible to choose at least a $1-\epsilon$ fraction of the codewords in $\Tilde{\cC}_{\mathrm{er}}$ such that \textit{any} $(2+\delta)\log n$ segment in \textit{any} of its codewords is distinct from \textit{any} $(2+\delta)\log n$ segment of the pilot sequence~$\mathbf{p}$.

Let $\cS=\{\bolds_1,\bolds_2,\dots,\bolds_{|\cS|}\}\subset\Tilde{\cC}_{\mathrm{er}}$ be a set of $(1-\epsilon)2^{\frac{n}{m}R_{\mathrm{er}}}$ such codewords, i.e., where none of them shares any $(2+\delta)\log n$ contiguous subsequence with $\mathbf{p}$. 
Ref.~\cite{shomorony2021torn} builds a codeword $\boldc$ by taking $m-1$ codewords 
from $\cS$ and interleaving their symbols with symbols from $\mathbf{p}$. 
More precisely, for any function~$u:[m-1]\to[|\cS|]$, they build the codeword $\boldc_u=(\boldc_u[0],\boldc_u[1],\ldots, \boldc_u[n-1])$ where for every $t\in\{0,1,\dots,n/m-1\}$, we have
\begin{equation}
    \boldc_u[mt+j]=
    \begin{cases}
    \mathbf{p}[t], & j= 0, \\
    \bolds_{u(j)}[t], & j=1, 2,3,\dots,m-1.
\end{cases}
\end{equation}
The resulting codebook is of size $|\cS|^{m-1}=(1-\epsilon)^{m-1}2^{(1-\frac{1}{m})nR_{\mathrm{er}}}$. 

Subsequently, decoding begins by recovering the global alignment of the fragments through the embedded pilot symbols drawn from the De Bruijn sequence $\mathbf{p}$. 
Since a window of length $(2+\delta)\log n$ from the pilot sequence coincides with a window of the shifted erasure code with probability which tends to zero as $n \to \infty$, every fragment whose length exceeds $(2+\delta)m\log n$ must contain at least one pilot window of length $(2+\delta)\log n$ that is uniquely identifiable. 
Moreover, since every $\log(n/m)$ window of the De Bruijn sequence $\mathbf{p}$ is distinct, the occurrence of such a window immediately reveals the fragment’s absolute position in the codeword, a process we call \textit{global alignment}. 
Once these long fragments are anchored, the short fragments—those too small to contain any identifiable pilot window—are discarded, and the resulting missing positions are treated as erasures. 
The reconstruction of the full message is then completed by the outer shifted-erasure code $\Tilde{\cC}_{\mathrm{er}}$, which deterministically corrects the erased locations created by the random breaking process.

Notice that in the interleaved-pilot scheme, the decoder identifies the locations of the pilot sequence bits only in sufficiently long fragments. 
The rationale is that the probability of a long window of the De Bruijn sequence coinciding with a window of any shifted erasure codeword is vanishingly small. 
However, this approach has an inherent gap: the number of bits required to ensure that a De Bruijn window is \emph{not} identical to a shifted erasure codeword window is approximately  $2\log n$ (Observation~\ref{Lemma: probability of match in ilan's method}), whereas the number of bits needed to uniquely identify a De Bruijn window is only $\log (n/m)$. 
Consequently, the decoder in~\cite{shomorony2021torn} restricts itself to fragments of length at least $(2+\delta)m\log n$, discarding the large number of shorter fragments, which reduces the effective rate since all such fragments are treated as erasures.
This observation motivates the question of whether one can identify the locations of the pilot sequence bits by exploiting the local pattern surrounding each De Bruijn bit, rather than relying solely on long fragments. 

\subsection{The local-alignment method by Liu and Raviv}\label{section:previous-local-alignment}

Our previous work~\cite{liu2026improved} improves the interleaved-pilot method by allowing the decoder to use much shorter fragments. 
Instead of comparing a long pilot window with all possible windows of shifted erasure code, the decoder looks for a short pattern that appears only in the pilot sequence.
To create this pattern, each shifted erasure codeword is first encoded so that it cannot contain a long run of zeros. 
Let $\beta$ be a positive integer such that~$\beta=o(\log(n))$ and~$\beta=\omega(1)$.
A simple Run Length Limited encoder ensures that $0$-runs of length~$\beta$ never appear in an erasure codeword.
The pilot is constructed from a De Bruijn sequence with all-zero pattern $0^\beta$ inserted at designated positions.
The final codeword is then formed in the same way as in the interleaved-pilot method: one pilot sequence is interleaved with $m-1$ erasure codewords.
Since $\beta=o(\log(n))$, the extra bits introduced by this RLL encoder form a vanishing fraction of the codeword and therefore cause only a vanishing rate loss.

When a fragment of length at least $(1+\delta)m\log n$ is received, the decoder separates it into the $m$ possible de-interleaved sequences. 
Since the marker cannot appear in an RLL-coded erasure code and the de-interleaved subsequence from the pilot sequence contains a complete marker, the decoder can identify which symbols of the fragment belong to the pilot simply by locating the marker. 
This step is called \emph{local alignment}, since the pilot is recognized inside the fragment rather than from a comparison with the erasure code.
After identifying the pilot subsequence, the decoder removes the inserted markers and finds the bits belonging to the De Bruijn sequence which then reveal the fragment's position in the original codeword. 
All fragments that can be positioned in this way are placed in their correct locations. 
Fragments that are too short to contain sufficiently many bits from the  pilot sequence are discarded, and are treated as erasures. 
The random binary linear code decoder is then used to correct these erasures.

The main advantage of local alignment is that it reduces the required fragment length threshold from approximately $(2+\delta)m\log n$ in the method of~\cite{shomorony2021torn} to approximately $(1+\delta)m\log n$. 
Consequently, many shorter fragments that were previously discarded can now be used for decoding. 
Although this method improves the achievable rate substantially, it still does not achieve capacity due to the following tradeoff. 
Using more pilot symbols allows the decoder to use shorter fragments but leaves fewer positions for information. 
Using fewer pilot symbols leaves more positions for information but causes more short fragments to be discarded. Therefore, even after choosing the optimal~$m$, the achievable rate remains below the torn-paper-channel capacity.

The construction proposed in the next section improves this method through \emph{successive local alignment and pilot recycling}. The original pilot sequence is used to decode the first information level.
Once that level has been recovered, its known symbols are reused as additional pilot symbols to align shorter fragments. 
Repeating this process allows the decoder to use progressively shorter fragments without inserting more fixed pilot symbols into the codeword.

\section{Successive local alignment with pilot recycling}\label{section:successive alignment}
We propose an improved coding scheme that builds on the local-alignment technique of~\cite{liu2026improved}. The main new ingredient is a multilevel pilot-recycling procedure. 
At the initialization level, a fixed pilot sequence is used to identify and globally align sufficiently long fragments. 
The information recovered from those fragments is then decoded and reused as an additional larger pilot sequence at the next level. 
Repeating this procedure allows the decoder to decode progressively shorter fragments. 
At deeper levels of the recursion process, a fragment is decoded only when it admits a unique global placement; fragments with more than one consistent placement are discarded. 
It is shown in Section~\ref{Section:rate-and-reliability} that the total normalized length of these ambiguous fragments vanishes, so their removal creates only a negligible number of additional erasures.

The construction has two components. 
First, as in~\cite{liu2026improved}, all-zero markers are embedded inside the pilot sequence at special locations, whereas every information-bearing level codeword is constrained by an $\mathrm{RLL}(0,\beta-1)$ encoder (see Definition~\ref{definition:random-linear-erasure-codes} for $\beta$). 
This makes the pilot sequences locally identifiable, as in~\cite{liu2026improved}.
Second, we partition the set of codeword positions into \textit{residue classes} according to certain modular conditions over their indices. 
Then, at each level, different residue classes are used to assign
erasure-coded information.
Once one level has been decoded, the corresponding decoded positions become known and enlarge the pilot sequence available to the next level.

We first define the levelized design;
subsequent Lemma~\ref{lemma:periodic-coordinate-properties} proves the properties used throughout the encoding, decoding, and rate analysis.
The following definitions set the fundamental relations among the parameters of our constructions; as shown in the sequel, different parameters will lead to different rates. 
\begin{definition}
\label{definition:alignment-scales}
Fix integers $T>1$ and $m>1$, and let
$
J\triangleq m(T-1).
$
Set $t_0\triangleq 1$ and, for every $i\in[J]$, define
$
t_i\triangleq 1+\frac{i}{m},
$
and write \(t_i=b_i/a_i\), with
\(a_i,b_i\in\mathbb N\), \(\gcd(a_i,b_i)=1\) for every~\(i\). 
Furthermore, let
\(Q\triangleq\operatorname{lcm}(b_1,\ldots,b_J)\), and let
\(q_i\triangleq Q/t_i=Qa_i/b_i\in\bN\).
\end{definition}
For convenience of exposition, we assume that~$Q$ in Definition~\ref{definition:alignment-scales} divides~$n$.
The parameter \(t_i\) specifies the alignment scale at level \(i\). More
precisely, for a fixed constant \(\delta>0\), the decoder at level \(i\)
aims to align every fragment of length at least
\((1+\delta)t_i\log n\). Since \(t_J=T\), the decoder begins at level~\(J\)
by aligning the longest fragments and then proceeds recursively to lower
levels. At each step, the bits recovered from the previously decoded levels
increase the density of known positions, thereby allowing progressively
shorter fragments to be aligned.
At the final level, the decoder aligns all fragments of length at least
\((1+\delta)t_1\log n\). By choosing \(m\) sufficiently large, $t_1$ can be made arbitrarily close to $1$ which approaches the threshold required for achieving
capacity, and hence the resulting achievable rate can be made arbitrarily
close to capacity.

We next define the residue classes used to assign codeword positions to
the fixed pilot sequence and to the different coding levels.
For example, if a residue
class is labeled by \(r\in\{0,1,\ldots,Q-1\}\), then it contains the
positions \(r,Q+r,2Q+r,\ldots\). 
During decoding, once the bits in one residue class have
been recovered, those positions become known and can be reused to help
align shorter fragments at the next level. 
The next definitions are illustrated in Example~\ref{example:nested-periodic-alignment-sets} which follows.

\begin{definition}
\label{definition:nested-residue-sets}
Let \(\cR_0\triangleq\bZ_Q=\{0,1,\ldots,Q-1\}\). Define the \emph{pilot
residue set} by
\(\cR_J\triangleq\{0,T,2T,\ldots,(q_J-1)T\}\). 
Write the residues outside~$\cR_J$ in increasing order as
\(\bZ_Q\setminus\cR_J
=\{u_0,u_1,\cdots,u_{Q-q_J-1}\}\), where~$u_0<u_1<\ldots<u_{Q-q_J-1}$. 
For every \(i\in[J-1]\), define
\(\cR_i\triangleq\cR_J\cup
\{u_0,u_1,\ldots,u_{q_i-q_J-1}\}\), and for $i\in[J]$ let \(\cB_i\triangleq\cR_{i-1}\setminus\cR_i\).
\end{definition}

We now extend these residue sets from one period of length \(Q\) to the full length-\(n\) codeword by collecting all codeword bits whose residues modulo \(Q\) belong to the corresponding sets.

\begin{definition}\label{definition: Si and Ai}
    Let \(\cS_0\triangleq\{0,1,\ldots,n-1\}\) and for every \(i\in[J]\), let
\(\cS_i\triangleq\{x\in\{0,1,\ldots,n-1\}:x\bmod Q\in\cR_i\}\).
Further, set \(\cA_i\triangleq\cS_{i-1}\setminus\cS_i=\{x\in\{0,1,\ldots,n-1\}:x\bmod Q\in\cB_i\}\). 
\end{definition}

By construction, the residue sets are nested as
\(\cR_J\subseteq\cR_{J-1}\subseteq\cdots\subseteq\cR_1\subseteq\cR_0\),
with \(|\cR_i|=q_i\). 
The set \(\cR_J\) is reserved for the fixed
pilot sequence, while the difference
\(\cB_i=\cR_{i-1}\setminus\cR_i\) is assigned to the level-\(i\)
information codeword using RLL-encoded random binary linear code.
During encoding, the level-\(i\) codeword is
placed at all codeword positions whose residues modulo \(Q\) belong to
\(\cB_i\), i.e. the set \(\cA_i\).

The same nested residue assignment determines the recursive decoding order. 
Initially, only the pilot bits in \(\cS_J\) are known to the decoder. 
After level \(J\) is decoded,
bits in \(\cA_J\) also become known, enlarging the set of known bits to
\(\cS_{J-1}=\cA_J\mathbin{\dot\cup}\cS_J\) (see Remark~\ref{remark: disjoint union} for $\dot\cup$). 
More generally, after levels
\(i+1,\ldots,J\) have been decoded, the decoder knows all the bits in
\(\cS_i\). 
Decoding level~\(i\) reveals the bits in
\(\cA_i\), so the set of known bits expands to
\(\cS_{i-1}=\cA_i\mathbin{\dot\cup}\cS_i\). 
Thus, each decoded level
provides a larger set of known bits for aligning progressively shorter fragments at the
next level.
The following Example~\ref{example:nested-periodic-alignment-sets} illustrates this nested residue assignment, together with the corresponding placement of the
pilot sequence and all the level codewords, and is shown in Figure~\ref{Figure: recursive alignment}.

\begin{remark}\label{remark: disjoint union}
    $\mathbin{\dot\cup} $ represents a disjoint union of sets, meaning the sets being combined share no elements in common.
\end{remark}
\begin{example}\label{example:nested-periodic-alignment-sets}
Let \(m=2\), \(T=3\), and \(t_1=3/2\). By
Definition~\ref{definition:alignment-scales}, we have
\(J=1+m(T-t_1)=4\), and therefore
\[
(t_0,t_1,t_2,t_3,t_4)
=
\left(1,
\frac32,\,
2,\,
\frac52,\,
3
\right).
\]
hence
\((b_1,b_2,b_3,b_4)=(3,2,5,3)\), so
\(
Q=\operatorname{lcm}(3,2,5,3)=30\).
Since \(q_i=Q/t_i\) and \(q_0\triangleq Q\), we obtain
\[
(q_0,q_1,q_2,q_3,q_4)=(30,20,15,12,10).
\]
By Definition~\ref{definition:nested-residue-sets} the pilot residue set is
\[
\cR_4
=
\{0,3,6,9,12,15,18,21,24,27\},
\]
the remaining residues are
\[
\bZ_{30}\setminus\cR_4
=
\{1,2,4,5,7,8,10,11,13,14,
16,17,19,20,22,23,25,26,28,29\}.
\]
By Definition~\ref{definition:nested-residue-sets}, the nested
residue sets are obtained by successively adding the required number of
nonpilot residues, i.e.,
\[
\begin{aligned}
\cR_3&=\cR_4\cup\{1,2\},
&\cB_4&=\{1,2\},\\
\cR_2&=\cR_4\cup\{1,2,4,5,7\},
&\cB_3&=\{4,5,7\},\\
\cR_1&=\cR_4\cup\{1,2,4,5,7,8,10,11,13,14\},
&\cB_2&=\{8,10,11,13,14\},\\
\cR_0&=\{0,1,\ldots,29\},
&\cB_1&=\{16,17,19,20,22,23,25,26,28,29\}.
\end{aligned}
\]
\end{example}
\begin{figure}[ht]
    \centering
    \includegraphics[width=0.9\linewidth]{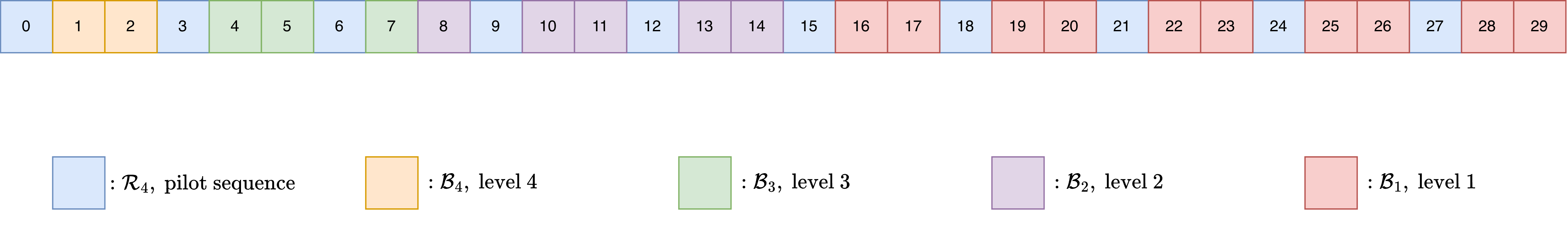}
    \caption{An illustration of the partition of~~$\bZ_Q$ to residue classes, given in Example~\ref{example:nested-periodic-alignment-sets}.}
    \label{Figure: recursive alignment}
\end{figure}
Definition~\ref{definition: Si and Ai} extends the local modular relations in Definition~\ref{definition:nested-residue-sets} to the entire codeword, as exemplified next.

\begin{example}\label{example:full-position-sets}
Following Example~\ref{example:nested-periodic-alignment-sets}, we have
\(Q=30\) and the residue sets
\(\cR_4,\ldots,\cR_0\) and
\(\cB_4,\ldots,\cB_1\).
We apply Definition~\ref{definition: Si and Ai} to obtain the
corresponding $\cS_i$'s and $\cA_i$'s in the length-\(n\) codeword.
We have
\[
\cS_4
=
\{30u+b:
u=0,1,\ldots,n/30-1,\;
b\in\{0,3,6,9,12,15,18,21,24,27\}\}=\{0,3,6,\ldots,n-3\},
\]
\[
\begin{aligned}
\cS_3
&=\cS_4\mathbin{\dot\cup}
\{30u+b:
u=0,\ldots,n/30-1,\;
b\in\cB_4=\{1,2\}\},
&
\\
\cA_4
&=
\{30u+b:
u=0,\ldots,n/30-1,\;
b\in\cB_4=\{1,2\}\},
\\
\cS_2
&=\cS_4\mathbin{\dot\cup}
\{30u+b:
u=0,\ldots,n/30-1,\;
b\in\cB_4\cup\cB_3=\{1,2,4,5,7\}\},
&
\\
\cA_3
&=
\{30u+b:
u=0,\ldots,n/30-1,\;
b\in\cB_3=\{4,5,7\}\},
\\
\cS_1
&=\cS_4\mathbin{\dot\cup}
\{30u+b:
u=0,\ldots,n/30-1,\;
b\in\cB_4\cup\cB_3\cup\cB_2=\{1,2,4,5,7,8,10,11,13,14\}\},
&
\\
\cA_2
&=
\{30u+b:
u=0,\ldots,n/30-1,\;
b\in\cB_2=\{8,10,11,13,14\}\},
\\
\cS_0
&=\{0,1,\ldots,n-1\},
&
\\
\cA_1
&=
\{30u+b:
u=0,\ldots,n/30-1,\;
b\in\cB_1=\{16,17,19,20,22,23,25,26,28,29\}\}.
\end{aligned}
\]
For example,
\[
\cA_4=\{1,2,31,32,61,62,\ldots\}
\mbox{ and }
\cA_3=\{4,5,7,34,35,37,64,65,67,\ldots\}.
\]
\end{example}

To put Definitions~\ref{definition:alignment-scales}-\ref{definition: Si and Ai} in context, we briefly explain the encoding/decoding process; the full details are in Section~\ref{Section:proposed-work encoding} and Section~\ref{Section:levelized-decoding}.
\begin{figure}[ht]
    \centering    \includegraphics[width=1.0\linewidth]{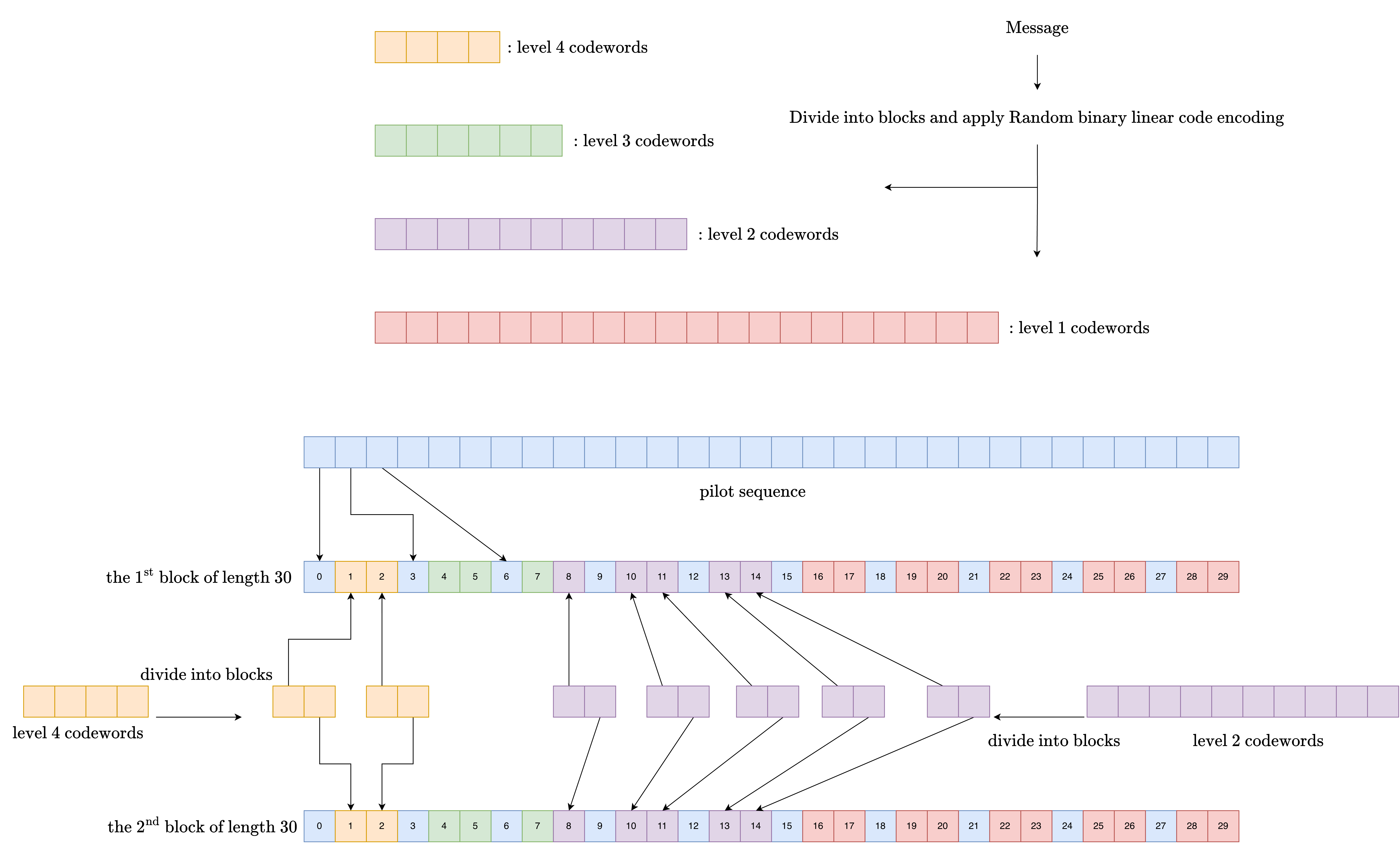}
    \caption{A sketch of encoding.}
    \label{Figure: encoding}
\end{figure}
Figure~\ref{Figure: encoding} illustrates the encoding process using the residue sets and codeword bits constructed in Examples~\ref{example:nested-periodic-alignment-sets} and~\ref{example:full-position-sets}. 
The message is first divided into four parts, and each part is encoded by the random binary linear code corresponding to one level, producing the level-$1$, level-$2$, level-$3$, and level-$4$ codewords. 
The fixed pilot sequence is generated separately and the codeword is then viewed as consecutive blocks of length $Q=30$.
Within each length-$30$ block, the pilot bits are placed at the positions whose residues belong to
\[
\cR_4=\{0,3,6,9,12,15,18,21,24,27\},
\]
whereas the four level codewords are placed according to
\[
\begin{aligned}
\cB_4&=\{1,2\},\\
\cB_3&=\{4,5,7\},\\
\cB_2&=\{8,10,11,13,14\},\\
\cB_1&=\{16,17,19,20,22,23,25,26,28,29\}.
\end{aligned}
\]
For every level $i$, the corresponding level codeword is self-interleaved across the residue classes in $\cB_i$. 
More precisely, the level-$i$ codeword is divided into $|\cB_i|$ equal consecutive blocks, and each block is assigned to one residue class in $\cB_i$. 
The bits of that block are then placed successively in the same residue class across consecutive length-$Q$ blocks of the codeword. In this example, the $|\cB_4|=2$ blocks of the level-$4$ codeword are placed in residues $1$ and $2$ modulo $30$, while the $|\cB_2|=5$ blocks of the level-$2$ codeword are placed in residues $8,10,11,13,$ and $14$ modulo $30$.
The pilot sequence is placed consecutively at positions $0,3,6,\ldots$. 
Hence, in every length-$30$ block, the pilot occupies the whole $\cR_4$, while each level codeword occupies exactly the positions corresponding to its set $\cB_i$. 
Repeating this placement over all length-$Q$ blocks fills the entire codeword, with every position carrying either one pilot bit or one bit from exactly one level codeword.
\begin{figure}[ht]
    \centering
    \includegraphics[width=1.0\linewidth]{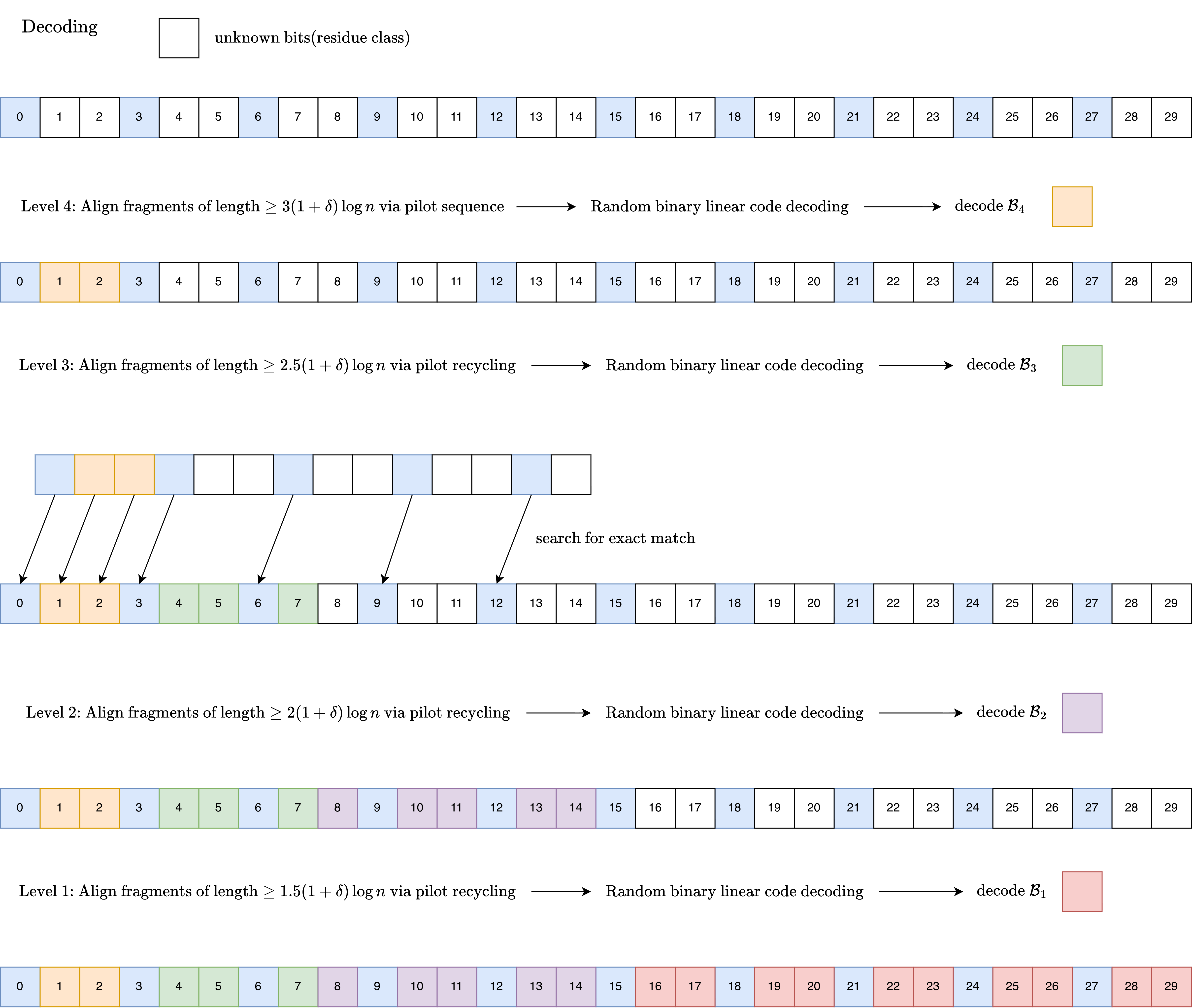}
    \caption{A sketch of decoding.}
    \label{Figure: decoding}
\end{figure}

We next provide an informal description for the decoding process.
Figure~\ref{Figure: decoding} shows how the general recursive decoding rule is applied to the concrete residue assignment given in Example~\ref{example:nested-periodic-alignment-sets} and Example~\ref{example:full-position-sets}. 
Initially, only the fixed pilot positions
in \(\cS_4\), i.e., the
positions whose residues belong to
\[
\cR_4=\{0,3,6,9,12,15,18,21,24,27\},
\]
are known. 
All positions belonging to
\(\cB_4,\cB_3,\cB_2,\cB_1\) are initially treated as unknown bits, represented by white squares in Figure~\ref{Figure: decoding}.
 Starting at level \(4\), the decoder first uses fragments of length at least
\(3(1+\delta)\log n\). 
As explained later in
Section~\ref{subsection:initial-pilot-decoding}, these fragments are
aligned using the fixed pilot sequence.
Once these fragments have been placed at
their correct positions, a random binary linear code decoder described in
Section~\ref{subsection:level-erasure-decoding} reveals the entire level-\(4\) codeword
carried by
$\cB_4=\{1,2\}$.
After level \(4\) is
decoded, all codeword bits on \(\cA_4\) are known.

At this point, the first pilot-recycling step occurs. The recovered
level-\(4\) bits are no longer unknown bits.
They can therefore be
used together with the original pilot bits when aligning the next set of
fragments. 
In this example, after level \(4\) is decoded,  the known residue classes increase from $10$ pilot residues in
\(\cR_4\) to $12$ residues in
$\cR_3
=
\cR_4\cup\{1,2\}.$
In the figure, the orange squares numbered \(1\) and \(2\) therefore change from
unknown to known and are reused as additional pilot bits.

The decoder then moves to level \(3\). 
It uses all positions in \(\cS_3\) to align
fragments of length at least \(2.5(1+\delta)\log n\), which are shorter than
the fragments used at level \(4\). 
More precisely, Algorithm~\ref{alg:recycled-pilot-alignment} tests possible
global starting positions of a fragment.
For a candidate starting position,
the decoder compares every fragment bit with bits in
\(\cS_3\) and checks whether it agrees with the already known codeword
bit at that position. 
This comparison is illustrated by the
``search for exact match'' part of Figure~\ref{Figure: decoding}. 
The true starting position
always passes the comparison since the torn-paper channel only cuts and
reorders the codeword and does not change its bits. 
If exactly one
candidate position agrees with all known bits, the fragment can be uniquely aligned at that position.
If several candidate positions agree with all known bits,
the fragment is declared ambiguous and is discarded.

After the uniquely aligned fragments have been placed, they provide
observations of the level-\(3\) codeword in
$
\cB_3=\{4,5,7\}.
$
The remaining unknown level-\(3\) bits are treated as erasures and are
recovered using the level-\(3\) random binary linear code. 
Once level \(3\)
has been decoded, the green squares in Figure~\ref{Figure: decoding} also become known.
The known set therefore expands from
\(\cS_3\) to
\(\cS_2=\cS_3\mathbin{\dot\cup}\cA_3,\) where \(q_2=15\) of the \(30\) residue classes are now known and can be
used for the next step.

At level \(2\),  the decoder considers fragments of length at
least \(2(1+\delta)\log n\). 
It applies the same pilot-recycling procedure and level \(2\) random binary linear code decoder to reveal all bits in $\cB_2=\{8,10,11,13,14\}$,
shown as purple squares in Figure~\ref{Figure: decoding}.
The known set expands to $\cS_1=\cS_2\mathbin{\dot\cup}\cA_2$,
which corresponds to
\[
\cR_1
=
\cR_2\cup\cB_2
=
\cR_4\cup
\{1,2,4,5,7,8,10,11,13,14\}.
\]

Finally, at level \(1\), the decoder can align fragments as
short as \(1.5(1+\delta)\log n\).
These fragments are aligned using all
known positions in \(\cS_1\). 
The aligned fragments reveal observations of
the remaining level-\(1\) codeword carried on
\[
\cB_1
=
\{16,17,19,20,22,23,25,26,28,29\}.
\]
After the level-\(1\) random binary linear code is decoded, these remaining positions also become known.
Hence
$
\cS_0=\cS_1\mathbin{\dot\cup}\cA_1
=\{0,1,\ldots,n-1\},
$
and, within one block,
$
\cR_0
=
\cR_1\cup\cB_1
=
\{0,1,\ldots,29\}.
$
This example shows the complete recursion
\[
\cR_4
\subset
\cR_3
\subset
\cR_2
\subset
\cR_1
\subset
\cR_0,
\]
corresponding to the successive decoding of levels \(4,3,2,1\).
Decoding each level reveals a new set of codeword bits.
These recovered bits are then recycled as additional pilot bits, making a larger
set of known bits. 
As a result, the minimum length of fragment that can be aligned in each level
decreases successively from
\(3(1+\delta)\log n\), to
\(2.5(1+\delta)\log n\), to
\(2(1+\delta)\log n\), and finally to
\(1.5(1+\delta)\log n\).
This example illustrates the main idea of the proposed construction:
the decoder repeatedly reuses recovered bits to align
progressively shorter fragments.

The subsequent Lemma~\ref{lemma:periodic-coordinate-properties} shows several properties of Definition~\ref{definition:alignment-scales}-\ref{definition: Si and Ai} that are used throughout encoding, decoding, and rate analysis.

\begin{lemma}\label{lemma:periodic-coordinate-properties}
The sets in Definition~\ref{definition:alignment-scales}-\ref{definition: Si and Ai} satisfy the following properties.
\begin{enumerate}
    \renewcommand{\labelenumi}{(\Alph{enumi})}
    \renewcommand{\theenumi}{(\Alph{enumi})}

    \item\label{lemma:A} $\cR_J\subset\cR_{J-1}\subset\cdots\subset\cR_1\subset\cR_0$ and
    $\cS_J\subset\cS_{J-1}\subset\cdots\subset\cS_1\subset\cS_0$.
    
    \item\label{lemma:B} For every $i\in[J]$, we have
    $|\cR_i|=q_i=Q/t_i$ and $|\cS_i|=n/t_i$.
    
    \item\label{lemma:C} The sets $\cA_1,\ldots,\cA_J,\cS_J$ form a partition of
    $\{0,1,\ldots,n-1\}$:
    \[
    \{0,1,\ldots,n-1\}
    =\cA_1\mathbin{\dot\cup}\cA_2\mathbin{\dot\cup}\cdots
    \mathbin{\dot\cup}\cA_J\mathbin{\dot\cup}\cS_J.
    \]
    
    \item\label{lemma:D} For every $i\in[J]$, we have
    \[
    \cA_i=\{uQ+b:u\in\{0,1,\ldots,n/Q-1\},\ b\in\cB_i\}
    \]
    and
    \[
    |\cA_i|
    =n\left(\frac{1}{t_{i-1}}-\frac{1}{t_i}\right)
    =\frac{n}{Q}|\cB_i|.
    \]
\end{enumerate}
\end{lemma}

\begin{proof}[Proof of~\ref{lemma:A}]
By Definition~\ref{definition:alignment-scales}, we have
\(t_0<t_1<\cdots<t_J=T\), and since \(q_i=Q/t_i\), it follows that
\(q_0>q_1>\cdots>q_J\). By
Definition~\ref{definition:nested-residue-sets}, for every
\(i\in[J-1]\), the set \(\cR_i\) consists of \(\cR_J\) together with the
first \(q_i-q_J\) elements of the ordered complement
\(\bZ_Q\setminus\cR_J\). Therefore, as \(q_i\) decreases with \(i\), the
sets satisfy
\(\cR_J\subset\cR_{J-1}\subset\cdots\subset\cR_1\subset\cR_0\).
By Definition~\ref{definition: Si and Ai}, the set \(\cS_i\) contains
exactly those positions whose residues modulo \(Q\) lie in \(\cR_i\), so
the same nesting relation gives
\(\cS_J\subset\cS_{J-1}\subset\cdots\subset\cS_1\subset\cS_0\).
\end{proof}

\begin{proof}[Proof of~\ref{lemma:B}]
The set \(\cR_J=\{0,T,2T,\ldots,(q_J-1)T\}\) contains exactly
\(q_J=Q/T\) residues. For every \(i\in[J-1]\), the set \(\cR_i\) is
obtained by adjoining exactly \(q_i-q_J\) residues to \(\cR_J\), and hence
\(|\cR_i|=q_J+(q_i-q_J)=q_i=Q/t_i\).

Since \(Q\mid n\), the interval \(\{0,1,\ldots,n-1\}\) consists of exactly
\(n/Q\) complete blocks of length \(Q\). In each such block, exactly
\(|\cR_i|=q_i\) positions have residues belonging to \(\cR_i\), and
therefore \(|\cS_i|=(n/Q)q_i=n/t_i\). 
\end{proof}

\begin{proof}[Proof of~\ref{lemma:C}]
By Definition~\ref{definition: Si and Ai},
\(\cA_i=\cS_{i-1}\setminus\cS_i\). Since~\ref{lemma:A} gives
\(\cS_i\subset\cS_{i-1}\), we have
\(\cS_{i-1}=\cA_i\mathbin{\dot\cup}\cS_i\) for every \(i\in[J]\).
Applying this identity recursively yields
\[
\cS_0
=\cA_1\mathbin{\dot\cup}\cA_2\mathbin{\dot\cup}\cdots
\mathbin{\dot\cup}\cA_J\mathbin{\dot\cup}\cS_J.
\]
Since \(\cS_0=\{0,1,\ldots,n-1\}\), the sets
\(\cA_1,\ldots,\cA_J,\cS_J\) form the stated disjoint partition. 
\end{proof}

\begin{proof}[Proof of~\ref{lemma:D}]
By Definition~\ref{definition: Si and Ai}, a position belongs to
\(\cA_i=\cS_{i-1}\setminus\cS_i\) exactly when its residue modulo \(Q\)
belongs to \(\cB_i=\cR_{i-1}\setminus\cR_i\). Since \(Q\mid n\), every
such position can be written uniquely as \(uQ+b\), where
\(u\in\{0,1,\ldots,n/Q-1\}\) and \(b\in\cB_i\). Hence
\[
\cA_i=\{uQ+b:u\in\{0,1,\ldots,n/Q-1\},\ b\in\cB_i\}.
\]

Each residue in \(\cB_i\) occurs exactly \(n/Q\) times in
\(\{0,1,\ldots,n-1\}\), so \(|\cA_i|=(n/Q)|\cB_i|\). Moreover,
\(|\cB_i|=|\cR_{i-1}|-|\cR_i|=q_{i-1}-q_i\). Using~\ref{lemma:B} and
\(q_j=Q/t_j\), we obtain
\[
|\cA_i|
=\frac{n}{Q}(q_{i-1}-q_i)
=n\left(\frac{1}{t_{i-1}}-\frac{1}{t_i}\right).
\]
Therefore,
\[
|\cA_i|
=n\left(\frac{1}{t_{i-1}}-\frac{1}{t_i}\right)
=\frac{n}{Q}|\cB_i|.\qedhere
\]
\end{proof}
Lemma~\ref{lemma:periodic-coordinate-properties} provides the
 structure used by both the encoder and the decoder. 
 In particular, each level-\(i\) codeword is placed on the set \(\cA_i\), while the nested sets \(\cS_i\) describe the bits that become known successively during decoding.
After the fragments used at level \(i\) have been globally
aligned, the uncovered positions of the level-\(i\) codeword are treated as erasures. 
These erasures are neither independent nor memoryless, since they
are induced by the fragmentation process. 
This motivates protecting the information assigned to each level by an independent random binary linear code.
\section{Encoding}\label{Section:proposed-work encoding}
We now describe the encoding procedure formally. 
Similar to~\cite{liu2026improved}, we apply Run Length Limited (RLL) constraints to distinguish the fixed pilot sequence from the information-bearing codewords during local alignment. 
Specifically, every level codeword is encoded using an \(\mathrm{RLL}(0,\beta-1)\) block encoder (see Lemma~\ref{lemma:rll-beta-minus-one} below), so that \(\beta\) consecutive zeros cannot occur within a level codeword. 
The pilot sequence is encoded separately and contains deliberately inserted all-zero markers of length \(q_J\beta\). 
The difference between the marker length \(q_J\beta\) and the RLL constraint of \(\beta\) consecutive zeros is explained in Lemma~\ref{lemma:initial-local-alignment-correct}. 
These markers therefore allow the decoder to identify the de-interleaved subsequence containing the pilot bits. 
Once the pilot subsequence has been identified, the De Bruijn bits contained in it are used to determine the global position of the fragment.

The two fixed-length RLL encoders used below are identical to the ones
introduced in~\cite[Lemmas~1 and~2]{liu2026improved}. 
We restate their properties here to keep the notation and construction self-contained; their proofs are included in Appendix~\ref{appendix:rll-encoders} for completeness.
\begin{lemma}\label{lemma:rll-beta-minus-one}
There exists a $\mathrm{RLL}(0,k-1)$ fixed-length block code with rate $(k-1)/k$.
\end{lemma}
\begin{proof}
    See Appendix~\ref{appendix:rll-encoders}.
\end{proof}
\begin{lemma}\label{lemma:rll-beta-minus-two}
There exists a $\widetilde{\mathrm{RLL}}(0,k-2)$ fixed-length block code with rate $(k-2)/k$.
\end{lemma}
\begin{proof}
    See Appendix~\ref{appendix:rll-encoders}.
\end{proof}
The difference between a $\mathrm{RLL}(0,k-1)$ fixed-length block code and a $\widetilde{\mathrm{RLL}}(0,k-2)$ fixed-length block code is that the $\mathrm{RLL}(0,k-1)$ encoder adds a $1$ after every $k-1$ input bits and the $\widetilde{\mathrm{RLL}}(0,k-2)$ encoder adds a $1$ at the beginning and the end of every input $k-2$ bits; both encoders are required in the sequel.

The encoding process begins with an inner erasure-correcting code.
At each level, once fragments are aligned, the remaining bits  in that level are considered as erasures, turning the torn paper channel into an erasure channel (see proofs in Section~\ref{Section:rate-and-reliability}).
To recover the missing bits, we protect the information assigned to each level using independent random binary linear codes. 
Decoding of such codes can be implemented in polynomial time using Gaussian elimination. 
We next define the parameters for these codes.

\begin{definition}\label{definition:random-linear-erasure-codes}
Fix sufficiently small constants $\delta,\eta>0$. 
Let $\beta=\beta(n)$ satisfy $\beta=\omega(1)$ and $\beta=o(\log n)$. 
For every $i\in[J]$, 
$
n_{\mathrm{er}}^i\triangleq\frac{\beta-1}{\beta}|\cA_i|$.
Let $r_i\triangleq\left(1-(\alpha(1+\delta)t_i+1)e^{-\alpha(1+\delta)t_i}+3\eta\right)n^i_{\mathrm{er}}$, choose $H_i\in\mathbb F_2^{r_i\times n_{\mathrm{er}}^i}$ with uniformly independent entries and define the level-$i$ random binary linear code by
$\cC_{\mathrm{er}}^i\triangleq\ker(H_i)=\{\boldu\in\mathbb F_2^{n_{\mathrm{er}}^i}:H_i\boldu=\mathbf 0\}.$

Independently of $H_i$ and of the choices made at the other levels, choose a random shift $\boldz^{\,i}\sim\operatorname{Unif}\!(\mathbb F_2^{n_{\mathrm{er}}^i})$, and
define the shifted level-$i$ code by
$\cC_{\mathrm{sh}}^i\triangleq\{\boldu\oplus\boldz^{\,i}:
\boldu\in\cC_{\mathrm{er}}^i\}$.
The random shift is a part of the code construction, and is known to the decoder. 
Applying the $\mathrm{RLL}(0,\beta-1)$ encoder from Lemma~\ref{lemma:rll-beta-minus-one} to every codeword in $\cC_{\mathrm{sh}}^i$ gives
$\bar{\cC}_{\mathrm{er}}^i\triangleq\{\operatorname{Enc}_{\mathrm{RLL}}(\boldv):\boldv\in\cC_{\mathrm{sh}}^i\}\subseteq\{0,1\}^{|\cA_i|}.$
We refer to $\cC_{\mathrm{er}}\triangleq\{\cC_{\mathrm{er}}^i:i\in[J]\}$, $\cC_{\mathrm{sh}}\triangleq\{\cC_{\mathrm{sh}}^i:i\in[J]\}$, and $\bar{\cC}_{\mathrm{er}}\triangleq\{\bar{\cC}_{\mathrm{er}}^i:i\in[J]\}$ as the level-wise meta-codes.
\end{definition}
The following lemma follows a standard argument, that the matrix $H_i$ achieves full row rank with high probability, thereby establishing the rate of $\cC_{\mathrm{er}}$.

\begin{lemma}\label{lemma:full rank random parity matrix}\cite[Lemmas~3]{liu2026improved}
Let $H_i\in\mathbb{F}_2^{r_i\times n^i_{\mathrm{er}}}$ be a random binary matrix
whose entries are chosen independently and uniformly from $\mathbb{F}_2$,
where $r_i\le n^i_{\mathrm{er}}$ and
$n^i_{\mathrm{er}}-r_i\to\infty$ as $n\to\infty$.
Then
\[
\Pr\bigl(\operatorname{rank}(H_i)<r_i\bigr)
<
2^{r_i-n^i_{\mathrm{er}}},
\]
and hence $\operatorname{rank}(H_i)=r_i$ with high probability.
Consequently, with high probability, the random binary linear code
$\cC^i_{\mathrm{er}}$ defined by the parity-check matrix $H_i$ has rate
$R^i_{\mathrm{er}}=1-\frac{r_i}{n^i_{\mathrm{er}}}$.
\\
In particular, if
$r_i=\left(1-(\alpha(1+\delta)t_i+1)e^{-\alpha(1+\delta)t_i}+3\eta\right)n^i_{\mathrm{er}}$
and
$(\alpha(1+\delta)t_i+1)e^{-\alpha(1+\delta)t_i}-3\eta>0$,
then
\[
n^i_{\mathrm{er}}-r_i
=
\left((\alpha(1+\delta)t_i+1)e^{-\alpha(1+\delta)t_i}-3\eta\right)
n^i_{\mathrm{er}}
\to\infty,
\]
and therefore
$R^i_{\mathrm{er}}=(\alpha(1+\delta)t_i+1)e^{-\alpha(1+\delta)t_i}-3\eta$
with high probability.
\end{lemma}

The codes $\{\cC_{\mathrm{er}}^i\}_{i=1}^{J}$ in Definition~\ref{definition:random-linear-erasure-codes} have the following properties.
When the message selects a codeword \(\boldu^{\,i}\in\cC_{\mathrm{er}}^i\), the encoder first forms \(\boldv^{\,i}\triangleq\boldu^{\,i}\oplus\boldz^{\,i}\) and then applies the \(\mathrm{RLL}(0,\beta-1)\) encoder. 
The shift makes bits of \(\boldv^{\,i}\) uniformly random, even when \(\boldu^{\,i}=\mathbf 0\), so that these bits can later be reused for recursive alignment; a similar shifting technique is used in~\cite{shomorony2021torn}. 
Since addition of \(\boldz^{\,i}\) is a bijection and the \(\mathrm{RLL}(0,\beta-1)\) encoder is injective, it follows that \(\left|\bar{\cC}_{\mathrm{er}}^i\right|=\left|\cC_{\mathrm{sh}}^i\right|=\left|\cC_{\mathrm{er}}^i\right|\). 
For all $i\in [J]$, as the block length for $\left|\bar{\cC}_{\mathrm{er}}^i\right|$ is $|\cA_i|$, the rate for $\left|\bar{\cC}_{\mathrm{er}}^i\right|$ goes to $(1-\frac{1}{\beta})R^i_{\mathrm{er}}$. 

We next construct the fixed pilot sequence.
By Definition~\ref{definition:nested-residue-sets}, a bit from the pilot sequence is placed every $T$ codeword bits. 
Equivalently, within every block of length $Q$, the pilot sequence occupies $q_J=Q/T$ residue classes in $\cR_J$.
When a fragment is de-interleaved into $T$ candidate subsequences, exactly one candidate subsequence is from the pilot sequences, while all the other candidates are from the information-bearing residue classes.
To make the candidate subsequence distinguishable, we use a marker of length $q_J\beta$ zeros: A run of \(q_J\beta\) zeros in an information-bearing candidate subsequence would create at least \(\beta\) consecutive zeros in one of the RLL encoded random binary linear codes.
The formal proof is given in Lemma~\ref{lemma:initial-local-alignment-correct} below.

Let \(b_n\triangleq(1+\delta/2)\log n\).
For the sake of simplicity, we consider \(b_n\), \(n/(Tb_n)\), and all
quantities in Definition~\ref{definition:random-linear-erasure-codes} as
integers.
We further assume that \(\beta\mid b_n\) and \(\beta Q\mid n\).
These standard integrality and divisibility assumptions can be removed
by padding \(o(n)\) bits, which would not change the asymptotic rate.
Let
\[
N_{q,n}\triangleq
\frac{n}{T}\left(1-\frac{q_J\beta}{b_n}\right)
\frac{\beta-2}{\beta}.
\]
Choose a cyclic binary De Bruijn sequence whose length is the smallest
power of two that is at least \(N_{q,n}\), and let \(\boldq\) consist
of \(N_{q,n}\) consecutive bits of this sequence.
The order of this De Bruijn sequence differs from
\(\log N_{q,n}\) (where $\log$ has base~$2$) by at most one. Therefore, all subsequent bounds
involving the De Bruijn order change only by an additive constant,
which does not affect the asymptotic analysis.

Apply the $\widetilde{\mathrm{RLL}}(0,\beta-2)$ encoder from Lemma~\ref{lemma:rll-beta-minus-two} to obtain \textit{modified De Bruijn sequence} $\bar{\boldq}$ of length $(n/T)(1-q_J\beta/b_n)$, as follows.

\begin{definition}\label{definition:modified-pilot}
For every $t\in\{0,1,\ldots,n/(Tb_n)-1\}$ and $j\in\{0,1,\ldots,b_n-1\}$, we define $\boldp=(\boldp[0],\boldp[1],\ldots,\boldp[n/T-1])$ by 
\[
\boldp[tb_n+j]\triangleq
\begin{cases}
0, & 0\leq j<q_J\beta,\\
\bar{\boldq}[t(b_n-q_J\beta)+(j-q_J\beta)], & q_J\beta\leq j<b_n.
\end{cases}
\]
\end{definition}

We now specify the placement of all level codewords and the pilot sequence into a codeword $\boldc=(\boldc[0],\boldc[1],\ldots,\boldc[n-1])$ in our code: 
\begin{enumerate}
 \item \emph{Encoding the message into level codewords.}
Divide the message $\boldm \in \mathbb{F}_2^{\sum_{i=1}^J n_{\mathrm{er}}^iR_{\mathrm{er}}^i}$ into \(J\) consecutive blocks
\[
\boldm
=
\boldm^{\,1}\mathbin{\|}\boldm^{\,2}
\mathbin{\|}\cdots\mathbin{\|}\boldm^{\,J},
\]
where, for every \(i\in[J]\),
\[
\boldm^{\,i}
\in
\mathbb F_2^{\,n_{\mathrm{er}}^iR_{\mathrm{er}}^i}.
\]
Encode each \(\boldm^{\,i}\) using the level-\(i\) random binary linear code
\(\cC_{\mathrm{er}}^i\) from
Definition~\ref{definition:random-linear-erasure-codes} to obtain
\(\boldu^{\,i}\in\cC_{\mathrm{er}}^i\).  Then apply the random shift
\(\boldz^{\,i}\) and define
\[
\boldv^{\,i}
\triangleq
\boldu^{\,i}\oplus\boldz^{\,i}
\in
\cC_{\mathrm{sh}}^i.
\]
\item \emph{Self-interleaving and placement of the level codewords.}
For every $i\in[J]$, write the set $\cB_i=\{b_{i,j}\}_{j=0}^{|\cB_i|-1}$ in increasing order, i.e., $b_{i,0}<b_{i,1}<\cdots<b_{i,|\cB_i|-1}$, where $|\cB_i|=q_{i-1}-q_i$. 
Let
\[
\bar{\bolds}^{\,i}
\triangleq
\operatorname{Enc}_{\mathrm{RLL}}(\boldv^{\,i})
=(\bar \bolds^{\,i}[0],\bar \bolds^{\,i}[1],\ldots,\bar \bolds^{\,i}[|\cA_i|-1])
\in\bar{\cC}_{\mathrm{er}}^i.
\]
Since $|\cA_i|=|\cB_i|(n/Q)$, partition $\bar{\bolds}^{\,i}$ into $|\cB_i|$ consecutive blocks of length $n/Q$:
\[
\bar{\bolds}^{\,i}=\bar{\bolds}^{\,i,0}\mathbin{\|}\bar{\bolds}^{\,i,1}\mathbin{\|}\cdots\mathbin{\|}\bar{\bolds}^{\,i,|\cB_i|-1}.
\]
For every $\ell\in\{0,1,\ldots,|\cB_i|-1\}$ and every $u\in\{0,1,\ldots,n/Q-1\}$, place the $u$-th bit of the $\ell$-th block at the position $uQ+b_{i,\ell}$, namely,
\[
\boldc[uQ+b_{i,\ell}]\triangleq\bar \bolds^{\,i,\ell}[u].
\]
Thus, the complete level-$i$ RLL codeword occupies exactly the entries indexed by~$\cA_i$.

\item \emph{Placement of the pilot.}
For every $v\in\{0,1,\ldots,n/T-1\}$, set $\boldc[vT]=\boldp[v]$ as in Definition~\ref{definition:modified-pilot}. Equivalently, writing $v=uq_J+\ell$ with $u\in\{0,1,\ldots,n/Q-1\}$ and $\ell\in\{0,1,\ldots,q_J-1\}$, set
\[
\boldc[uQ+\ell T]\triangleq \boldp[uq_J+\ell].
\]
Thus, the pilot sequence bits occupy positions $\{0,T,2T,\ldots,n-T\}$.
\end{enumerate}

By Lemma~\ref{lemma:periodic-coordinate-properties}, the sets $\cA_1,\ldots,\cA_J,\cS_J$ form a disjoint partition of $\{0,1,\ldots,n-1\}$, and hence the three steps assign exactly one bit to every position of $\boldc=(\boldc[0],\boldc[1],\ldots,\boldc[n-1])$.

\section{Levelized Decoding}\label{Section:levelized-decoding}
The decoder recursively transforms the torn-paper channel output into a sequence of erasure-decoding problems, beginning with the longest fragments and gradually incorporating shorter ones. 
At the first level \(J\), the decoder considers every sufficiently long fragment and de-interleaves it into \(T\) candidate subsequences. 
The unique subsequence containing the marker (a string of $q_J\beta$ zeros) is identified as the pilot subsequence. 
After removing the marker and the RLL-added $1$'s, the decoder extracts a sufficiently long subsequence of the original De Bruijn sequence. 
Since this subsequence is unique inside the De Bruijn sequence, it determines the fragment's global position.
The entire fragment can then be placed at its correct coordinates in the codeword.

The globally aligned fragments provide a partially observed version of the level-\(J\) codeword. 
Uncovered coordinates of the level-$J$ codeword are treated as erasures, which are corrected using the erasure correction capability of~$\bar{\cC}_{\mathrm{er}}^J$.
Once level-\(J\) has been decoded, the codeword bits in \(\cA_J\) become known in addition to the original pilot bits in \(\cS_J\).
Thus, the set of known bits expands to \(\cS_{J-1}=\cA_J\mathbin{\dot\cup}\cS_J\).

At every subsequent level \(i\), assuming that levels \(i+1,\ldots,J\) have already been decoded, all codeword bits indexed by \(\cS_i\) are known.
These recovered bits are reused as a larger pilot sequence for aligning fragments that may be too short to contain enough bits from  the fixed De Bruijn pilot sequence.
For each such fragment, the decoder tests all possible global
 positions and retains the fragment only when exactly one candidate position is consistent with every known bit of \(\cS_i\) appearing in the fragment. 
If more than one candidate position is consistent, the fragment is
declared ambiguous and discarded.
These aligned fragments yield a partially observed version of the level-\(i\) codeword placed on \(\cA_i\). 
The remaining uncovered bits are again treated as erasures and recovered using the erasure correction capability of~$\bar{\cC}_{\mathrm{er}}^i$.
After level \(i\) is decoded, its codeword bits become
known, and the set of known bits expands according to
\(\cS_{i-1}=\cA_i\mathbin{\dot\cup}\cS_i\). 
The decoder then repeats the same procedure at level \(i-1\), using the larger set \(\cS_{i-1}\) to align
progressively shorter fragments.

We now describe each stage of the levelized decoder in detail. 
We begin with level \(J\), where no previously decoded level codewords are available and the decoder relies on the fixed pilot sequence for local
alignment, as in~\cite{liu2026improved}.  
This provides the first step of the recursive decoding
procedure. 
We then describe the subsequent levels, where the bits recovered
at the preceding levels are recycled together with the fixed pilot sequence to
align progressively shorter fragments.  
The correctness and reliability of these alignment and erasure-decoding steps are established alongside the corresponding decoding procedures with technical details shown in Section~\ref{Section:rate-and-reliability}.

\begin{remark}\label{remark:local alignment}
    We remark that the decoding procedure roughly follows the \emph{local alignment} technique introduced in our earlier work~\cite{liu2026improved}, where the zero marker identifies the pilot sequence bits within one fragment, rather than relying exclusively on global statistics. 
    The feature enabled the rate improvement of~\cite{liu2026improved} over~\cite{shomorony2021torn}.
\end{remark}

\subsection{First level of decoding}
\label{subsection:initial-pilot-decoding}
Let \(\cF_J\) be the set of received fragments whose lengths are at least \((1+\delta)T\log n\).
\begin{remark}\label{remark:prefix of fragment}
    For every \(\boldf\in\cF_J\), only its first \((1+\delta)T\log n\) bits are used to determine its global location, hence to avoid clutter we abuse the notation and denote the \((1+\delta)T\log n\) prefix of a fragment by~$\boldf$.
Once the location for the prefix is known, the entire fragment can be placed globally.
\end{remark}

The local-alignment algorithm in~\cite[Algorithm~1]{liu2026improved}
de-interleaves a received fragment into candidate subsequences according to the interleaving algorithm and identifies the unique candidate subsequence that contains the marker as the pilot subsequence (this is \textit{local alignment}, see Remark~\ref{remark:local alignment}).
In the present construction, the pilot sequence bits appear exactly once every \(T\) codeword bits, so Algorithm~\ref{alg:initial-local-alignment} de-interleaves the fragment into \(T\) candidate subsequences by placing bits with the same index modulo \(T\) in the same subsequence. 
The candidate subsequence containing \(0^{q_J\beta}\) is identified as the pilot subsequence, and the correctness of
Algorithm~\ref{alg:initial-local-alignment} is proved in
Lemma~\ref{lemma:initial-local-alignment-correct}.
\begin{algorithm}[H]
\caption{Initial local alignment by marker detection}
\label{alg:initial-local-alignment}
\begin{algorithmic}[1]
\STATE \textbf{Input:} A fragment prefix
\(\boldf=(\boldf[0],\ldots,\boldf[(1+\delta)T\log n-1])\) of length
\((1+\delta)T\log n\).
\STATE \textbf{Output:} The index \(j_{\boldf}\in\bZ_T\) of the
pilot-containing de-interleaved subsequence, the pilot subsequence
\(\boldp_{\boldf}\) and de-interleaved subsequences
\(\boldy_{\boldf,r}\) for \(r\in\bZ_T\).
\STATE Initialize empty sequences
\(\boldw_{\boldf,0},\ldots,\boldw_{\boldf,T-1}\).
\FOR{\(j=0\) to \((1+\delta)T\log n-1\)}
    \STATE Set \(r\gets j\bmod T\).
    \STATE Append \(\boldf[j]\) to \(\boldw_{\boldf,r}\).
\ENDFOR
\STATE Initialize \(\cP_{\boldf}\gets\varnothing\).
\FOR{\(r=0\) to \(T-1\)}
    \IF{\(\boldw_{\boldf,r}\) contains a run \(0^{q_J\beta}\)}\label{line:marker test}
        \STATE Set
        \(\cP_{\boldf}\gets\cP_{\boldf}\cup\{r\}\).
    \ENDIF
\ENDFOR
\STATE Let \(j_{\boldf}\) be the unique element of \(\cP_{\boldf}\).
\STATE Set
\(\boldp_{\boldf}\gets\boldw_{\boldf,j_{\boldf}}\).
\FOR{\(r=0\) to \(T-1\)}
    \STATE Set
    \(\boldy_{\boldf,r}
    \gets\boldw_{\boldf,(j_{\boldf}+r)\bmod T}\).
\ENDFOR
\RETURN
\(j_{\boldf},\boldp_{\boldf},
\{\boldy_{\boldf,r}:r\in\bZ_T\}\).
\end{algorithmic}
\end{algorithm}

\begin{lemma}\label{lemma:initial-local-alignment-correct}
For every \(\boldf\in\cF_J\),
Algorithm~\ref{alg:initial-local-alignment} identifies the unique
de-interleaved subsequence that contains bits of the pilot sequence.
Consequently, \(\boldp_{\boldf}\) is exactly the pilot subsequence contained in \(\boldf\), and every \(\boldy_{\boldf,r}\) is the corresponding modulo-\(T\) de-interleaved subsequence of \(\boldf\).
\end{lemma}
\begin{proof}
See Lemma~\ref{lemma:initial-local-alignment-correct-appendix} in Appendix~\ref{appendix:rll-encoders}.
\end{proof}

After Algorithm~\ref{alg:initial-local-alignment}, the decoder has identified
the pilot subsequence \(\boldp_{\boldf}\) contained in \(\boldf\). 
Algorithm~\ref{alg:extract-initial-debruijn} removes the inserted markers and the two \(1\)'s added to each
length-\(\beta\) block by the
\(\widetilde{\mathrm{RLL}}(0,\beta-2)\) encoder. The resulting sequence is
a subsequence \(\boldq_{\boldf}\) of the original De Bruijn sequence \(\boldq\). 

\begin{algorithm}[H]
\caption{Extraction of a De Bruijn subsequence}
\label{alg:extract-initial-debruijn}
\begin{algorithmic}[1]
\STATE \textbf{Input:} The pilot subsequence
\(\boldp_{\boldf}=(\boldp_{\boldf}[0],\ldots,\boldp_{\boldf}[|\boldp_{\boldf}|-1])\).
\STATE \textbf{Output:} The recovered De Bruijn subsequence
\(\boldq_{\boldf}\).
\STATE Let \(\cU_{\boldf}\) be the set of start indices of all runs
\(0^{q_J\beta}\) in \(\boldp_{\boldf}\), and set
\(u^\star\gets\min\cU_{\boldf}\).
\STATE Initialize \(\boldq_{\boldf}\) as an empty sequence.
\FOR{\(j=0\) to \(|\boldp_{\boldf}|-1\)}
    \STATE Set \(v_j\gets(j-u^\star)\bmod b_n\) and
    \(w_j\gets(j-u^\star)\bmod\beta\)
    \IF{\(v_j\geq q_J\beta\) and \(w_j\notin\{0,\beta-1\}\)}\label{line:remove markers and rll added bits}
        \STATE Append \(\boldp_{\boldf}[j]\) to \(\boldq_{\boldf}\).
    \ENDIF
\ENDFOR
\RETURN \(\boldq_{\boldf}\).
\end{algorithmic}
\end{algorithm}
While constructing \(\boldq_{\boldf}\), the decoder also retains the
index of the fragment bit corresponding to its first retained De Bruijn
bit. This index is determined directly from \(j_{\boldf}\) returned by
Algorithm~\ref{alg:initial-local-alignment} and the marker/RLL positions identified in Algorithm~\ref{alg:extract-initial-debruijn}.
The following lemma is the counterpart of
\cite[Lemma~7]{liu2026improved}. It proves that  Algorithm~\ref{alg:extract-initial-debruijn}
correctly removes the markers and $1$'s added by \(\widetilde{\mathrm{RLL}}(0,\beta-2)\) encoder.

\begin{lemma}\label{lemma:extract-initial-debruijn-correct}
For every \(\boldf\in\cF_J\),
Algorithm~\ref{alg:extract-initial-debruijn} outputs 
subsequence \(\boldq_{\boldf}\) of the original De Bruijn sequence
\(\boldq\) contained in the pilot subsequence \(\boldp_{\boldf}\).
\end{lemma}

\begin{proof}
We first show that \(\cU_{\boldf}\), the set of starting positions of all
runs \(0^{q_J\beta}\) in \(\boldp_{\boldf}\), is exactly the set of
starting positions of the complete markers contained in
\(\boldp_{\boldf}\).
Recall from Definition~\ref{definition:modified-pilot} that \(\bar{\boldq}\) is obtained by applying the
\(\widetilde{\mathrm{RLL}}(0,\beta-2)\) encoder to \(\boldq\). Every
length-\(\beta\) block of \(\bar{\boldq}\) begins and ends with \(1\), and
therefore \(\bar{\boldq}\) cannot contain \(q_J\beta\) consecutive zeros.
Moreover, the markers \(0^{q_J\beta}\) are inserted between complete
length-\(\beta\) blocks of \(\bar{\boldq}\). The bit immediately before
each marker is the final \(1\) of an RLL block, and the bit immediately after
each marker is the initial \(1\) of the next RLL block. Consequently, any
length-\(q_J\beta\) zero run in \(\boldp\) is one of the inserted markers:
a window lying entirely in \(\bar{\boldq}\) cannot be all zero, and a window
overlapping a marker without coinciding with it contains one of the adjacent
\(1\)'s. It follows that every index in \(\cU_{\boldf}\) is the start of an
inserted marker, and hence \(u^\star\) is the start of a  marker.

Next, we demonstrate that Line~\ref{line:remove markers and rll added bits} correctly filters out the inserted markers alongside all bits introduced by the RLL encoder.
Since the starting indices of all markers are separated by exact multiples of $b_n$, it follows that $v_j$ satisfies $v_j < q_J\beta$ if and only if the bit $\boldp_{\boldf}[j]$ is part of an inserted marker.
Consequently, enforcing the condition $v_j \geq q_J\beta$ isolates and removes exactly the marker bits.
It then remains to remove the two boundary $1$'s inserted into each information block by the $\widetilde{\mathrm{RLL}}(0,\beta-2)$ encoder. 
By construction, the parameters $q_J\beta$, $b_n$, and their difference $b_n - q_J\beta$ are all integer multiples of $\beta$. 
Since the RLL blocks are perfectly aligned with this period, the $1^{\text{st}}$ and the $\beta^{\text{th}}$ bits immediately following any marker correspond exactly to the $1$'s appended by the encoder. 
Therefore, applying the modulo constraint $w_j \notin \{0, \beta-1\}$ precisely targets and removes these RLL-added bits, yielding the recovered De Bruijn subsequence.
Therefore, Algorithm~\ref{alg:extract-initial-debruijn}
correctly outputs \(\boldq_{\boldf}\).
\end{proof}

The following lemma is the counterpart of
\cite[Lemma~8]{liu2026improved}. 
It states explicitly that the recovered subsequence \(\boldq_{\boldf}\)
occurs uniquely in \(\boldq\) and can therefore be used for global
alignment.

\begin{lemma}\label{lemma:initial-debruijn-long-enough}
For every \(\boldf\in\cF_J\), the output of
Algorithm~\ref{alg:extract-initial-debruijn} satisfies
\(|\boldq_{\boldf}|\geq\log N_{q,n}\) for all sufficiently large \(n\).
Consequently, \(\boldq_{\boldf}\) occurs at a unique position in
\(\boldq\) and can be used to determine the global position of
\(\boldf\).
\end{lemma}

\begin{proof}
First, observe that
\(|\boldq_{\boldf}|<|\boldp_{\boldf}|=(1+\delta)\log n<2b_n=(2+\delta)\log n\) (see Remark~\ref{remark:prefix of fragment}). 
As only one marker occurs every $b_n$ bits,
\(\boldp_{\boldf}\)  intersects at most three markers (one in the middle and two at the ends).
Thus, Algorithm~\ref{alg:extract-initial-debruijn} removes at most
\(3q_J\beta\) marker bits.
Among the remaining \((1+\delta)\log n-3q_J\beta\) bits from
\(\bar{\boldq}\), Algorithm~\ref{alg:extract-initial-debruijn} removes the $1$'s added by $\widetilde{\text{RLL}}(0,\beta-2)$ encoder. 
Thus
\begin{align*}
    |\boldq_{\boldf}|&\geq((\beta-2)/\beta)((1+\delta/2)\log n-3q_J\beta)-2\beta\\
&=(1+\delta/2)\log n-(2(1+\delta/2)/\beta)\log n-3q_J\beta+6q_J-2\beta\\
&\overset{(\dagger)}{=}(1+\delta/2)\log n-o(\log n),
\end{align*}
where~$(\dagger)$ follows since $\beta=\omega(1)$ and $\beta=o(\log n)$.
On the other hand,
$$\log N_{q,n}=\log ((n/T)(1-q_J\beta/b_n)(\beta-2)/\beta)
=\log n+O(1).$$
 We
obtain
$$|\boldq_{\boldf}|-\log N_{q,n}
=\delta\log n/2-o(\log n)\overset{(\star)}{>}0,$$
where~$(\star)$ holds for all sufficiently large \(n\).

By Lemma~\ref{lemma:extract-initial-debruijn-correct},
\(\boldq_{\boldf}\) is a subsequence of \(\boldq\). Since
\(\boldq\) consists a De Bruijn sequence of order \(\log N_{q,n}\), every
binary sequence of length \(\log N_{q,n}\) occurs at most once in
\(\boldq\) (Definition~\ref{Definition:De Bruijn sequence}). 
Therefore, the unique occurrence of \(\boldq_{\boldf}\) in \(\boldq\)
determines the global codeword position of its first retained De Bruijn
bit. 
Since the decoder also knows which position of \(\boldf\) produced
this bit, its true starting position is obtained by subtracting the
corresponding fragment index from this global position.
\end{proof}
Having established that every fragment of length at least $(1+\delta)T\log n$ can be uniquely aligned to its exact global location, we now transition from the first level to the erasure decoding, and align fragments of shorter length in subsequent levels.

\subsection{Level-wise inverse-RLL and erasure decoding}\label{subsection:level-erasure-decoding}
In Section~\ref{subsection:initial-pilot-decoding}, we adapted the local
alignment method of~\cite{liu2026improved} to identify the pilot
subsequence of every fragment of length at least $(1+\delta)T\log n$ and  determine the fragment's global position. 
After all such fragments have been globally aligned, the decoder obtains a partially observed codeword in which the covered coordinates are known and the uncovered coordinates are treated as erasures.
Let the partially observed codeword be \(\widehat{\boldc}^{(J)}\), which is an estimated codeword for $\boldc$ at level $J$.
At every subsequent level \(i\), the partially observed codeword \(\widehat{\boldc}^{(i)}\) is obtained after the fragments have been aligned.
The decoder must then extract the level-\(i\) bits from its estimated codeword \(\widehat{\boldc}^{(i)}\), reverse the RLL encoding, and recover its erased bits for level-\(i\).

The following algorithm is the multilevel counterpart of
\cite[Algorithm~3]{liu2026improved}.
In~\cite{liu2026improved}, each RLL-encoded information codeword occupies a single residue class and can therefore be extracted directly from that class. In the present construction, the level-\(i\) RLL-encoded codeword is self-interleaved across all residue classes in \(\cB_i\). 
The decoder must consequently extract the corresponding residue-class subsequences and concatenate them in the same order used by the
encoder. 
It then applies the RLL decoder and removes the random shift
\(\boldz^{\,i}\) introduced in Definition~\ref{definition:random-linear-erasure-codes}. 
The result is a partially observed codeword of \(\cC_{\mathrm{er}}^i\), whose unknown coordinates are recovered by decoding the random binary linear code. 
The correctness of Algorithm~\ref{alg:level-rll} is established in
Lemma~\ref{lemma:level-rll-correct} which follows.
\begin{algorithm}[H]
\caption{Extraction and inverse-RLL decoding of level $i$}
\label{alg:level-rll}
\begin{algorithmic}[1]
\STATE \textbf{Input:} A level index $i\in[J]$, a partially observed codeword $\widehat{\boldc}^{(i)}\in\{0,1,\star\}^n$, $\cB_i=\{b_{i,j}\}_{j=0}^{|\cB_i|-1}$ (with $b_{i,0}<b_{i,1}<\cdots<b_{i,|\cB_i|-1}$), and the shift $\boldz^{\,i}$.
\STATE \textbf{Output:} A partially observed codeword $\widehat{\boldu}^{\,i}\in\{0,1,\star\}^{n_{\mathrm{er}}^i}$ from $\cC_{\mathrm{er}}^i$.
\FOR{$\ell=0$ to $|\cB_i|-1$}\label{line:reverse rll for information}
    \STATE Set $\widehat{\bolds}^{\,i,\ell}\gets(\widehat \boldc^{(i)}[b_{i,\ell}],\widehat \boldc^{(i)}[Q+b_{i,\ell}],\ldots,\widehat \boldc^{(i)}[(n/Q-1)Q+b_{i,\ell}])$.
\ENDFOR\label{line:reverse rll end}
\STATE Concatenate $\widehat{\bar{\bolds}}^{\,i}\gets\widehat{\bolds}^{\,i,0}\mathbin{\|}\cdots\mathbin{\|}\widehat{\bolds}^{\,i,|\cB_i|-1}$.
\STATE Initialize an empty sequence $\widehat{\boldv}^{\,i}$.
\FOR{$a=0$ to $|\cA_i|-1$}
    \IF{$(a+1)\bmod\beta\neq0$}\label{line:remove added 1's}
        \STATE Append $\widehat{\bar \bolds}^{\,i}[a]$ to $\widehat{\boldv}^{\,i}$.
    \ENDIF
\ENDFOR
\STATE Set $\widehat{\boldu}^{\,i}\gets\widehat{\boldv}^{\,i}\oplus\boldz^{\,i}$, where $\star\oplus b\triangleq\star$.
\RETURN $\widehat{\boldu}^{\,i}$.
\end{algorithmic}
\end{algorithm}
\begin{lemma}\label{lemma:level-rll-correct}
 Algorithm~\ref{alg:level-rll} outputs $\widehat{\boldu}^{\,i}$, a partially observed version of the codeword $\boldu^{\,i}\in\cC_{\mathrm{er}}^i$.
\end{lemma}

\begin{proof}
During encoding, the shifted codeword $\boldu^{\,i}\oplus\boldz^{\,i}$ is mapped to the RLL-constrained codeword \(\bar{\bolds}^{\,i}\in\bar{\cC}_{\mathrm{er}}^i\).
This codeword is divided into the consecutive blocks $\bar{\bolds}^{\,i,0},\ldots,\bar{\bolds}^{\,i,|\cB_i|-1}$, and block $\ell$ is placed at positions $uQ+b_{i,\ell}$. 
This is a bijection, and Lines~\ref{line:reverse rll for information}-\ref{line:reverse rll end} of Algorithm~\ref{alg:level-rll} invert this placement and reconstruct the partially observed codeword $\widehat{\bar{\bolds}}^{\,i}$ in its original order.
Then, we remove the  $1$'s added by the $\mathrm{RLL}(0,\beta-1)$ encoder
which appends one $1$ after every $\beta-1$ bits.
Since the encoder appends exactly one \(1\) after every \(\beta-1\) bits, these padding bits appear at indices \(\beta-1, 2\beta-1, \ldots\) within the reconstructed sequence.
Line~\ref{line:remove added 1's} removes these specific positions, leaving a partially observed version of the shifted codeword \(\boldu^{\,i}\oplus\boldz^{\,i}\), which completes the proof.
\end{proof}
Lemma~\ref{lemma:level-rll-correct} establishes that
Algorithm~\ref{alg:level-rll} correctly outputs
\(\widehat{\boldu}^{\,i}\), a partially observed version of the codeword
\(\boldu^{\,i}\in\cC_{\mathrm{er}}^i\).
The remaining task is to show that the erasure pattern in \(\widehat{\boldu}^{\,i}\) lies,
with high probability, within the correction capability of the level-\(i\)
random binary linear code defined in
Definition~\ref{definition:random-linear-erasure-codes}. 
Since this probabilistic analysis contains heavily technical details, we defer
it to Section~\ref{subsection:erasure-reliability}.

Proceeding under the guarantee that these erasures are corrected successfully with high probability, we now turn to the central iterative mechanism of our scheme: recursively using fully recovered level-\(i\) bits to assist in the alignment of progressively shorter fragments.

\subsection{Alignment with recycled pilot bits}\label{subsection:recycled-pilot-decoding}
Section~\ref{subsection:level-erasure-decoding} details the specific extraction and erasure correction executed within a single level. 
Once the fragments used at level \(i\) have been placed at their
correct global positions, Algorithm~\ref{alg:level-rll} reverses the
self-interleaving and the RLL encoding, and then removes the  shift
\(\boldz^{\,i}\). 
This produces a partially observed version of the
level-\(i\) codeword \(\boldu^{\,i}\in\cC_{\mathrm{er}}^i\). 
The remaining missing bits are treated as erasures and recovered using the corresponding random binary linear code.
We now explain how the bits recovered in this
step are reused to align shorter fragments at the next level.

The base case of the recursion is level \(J\), which is handled in
Section~\ref{subsection:initial-pilot-decoding}. 
At this level, the decoder uses the fixed pilot sequence on \(\cS_J\) to align every fragment of length at least \((1+\delta)T\log n\). 
More precisely, local alignment identifies the de-interleaved subsequence, which belongs to the pilot sequence, and the extracted De Bruijn
subsequence determines the fragment's global position. 
The aligned fragments produce a partially observed version of the level-\(J\) codeword, which is then
recovered by decoding the level-\(J\) codeword. 
Consequently, the decoder learns the codeword bits on \(\cA_J\), and the set of aligned bits expands from \(\cS_J\) to \(\cS_{J-1}=\cA_J\mathbin{\dot\cup}\cS_J\).

We now describe the recursive step. Fix \(i\in[J-1]\) and suppose that
levels \(i+1,\ldots,J\) have been decoded correctly. 
The decoder then knows the fixed pilot bits on \(\cS_J\), together with the codeword bits on \(\cA_{i+1},\ldots,\cA_J\) recovered from the previously decoded levels.
Hence, the decoder knows every codeword bit indexed by \(\cS_i=\cA_{i+1}\mathbin{\dot\cup}\cdots\mathbin{\dot\cup}\cA_J\mathbin{\dot\cup}\cS_J\). 
In other words, we have that \(\cS_i=\cA_{i+1}\mathbin{\dot\cup}\cS_{i+1}\), and hence, decoding level
\(i+1\) enlarges the set of aligned bits from \(\cS_{i+1}\) to \(\cS_i\). 
These newly recovered bits are then reused as additional pilot bits when aligning fragments at level \(i\), as explained next.

At level \(i\), the decoder considers every fragment of length at least
\((1+\delta)t_i\log n\) and uses its first \((1+\delta)t_i\log n\) bits to test all possible global starting positions (in the spirit of Remark~\ref{remark:prefix of fragment}). 
A candidate position is called \textit{consistent} if every fragment bit
that would fall on a known position in \(\cS_i\) agrees with the
corresponding known codeword bit. 
Since the torn-paper channel only cuts and reorders the codeword and introduces no bit errors, it follows that the true starting position of the fragment is always consistent.
If there is exactly one consistent position, it must be the true starting position, and the decoder places the entire fragment there. 
If there is more than one consistent position, the decoder cannot determine the true placement uniquely; the fragment is then declared ambiguous and discarded.

After fragments that can be uniquely aligned at level \(i\) have been placed, the decoder obtains a partially observed version of the level-\(i\) word carried
on \(\cA_i\).
It then applies Algorithm~\ref{alg:level-rll} and the level-\(i\) decoder to recover \(\boldu^{\,i}\in\cC_{\mathrm{er}}^i\). 
This reveals the codeword bits on \(\cA_i\), so the known set expands according to \(\cS_{i-1}=\cA_i\mathbin{\dot\cup}\cS_i\). 
The same procedure is then repeated at level \(i-1\), where the set \(\cS_{i-1}\) allows the decoder to process fragments as short as \((1+\delta)t_{i-1}\log n\).
Section~\ref{subsection:coverage-and-ambiguity} proves that the total
length of the ambiguous fragments discarded during this recursion will be $o(n)$.
Thus, the presence of ambiguous fragments introduces only a negligible number of additional erasures.
The fact that these erasures can be corrected by the binary linear codes, whose parameters were set in Definition~\ref{definition:random-linear-erasure-codes}, is proved in Section~\ref{subsection:erasure-reliability}.

\begin{algorithm}[H]
\caption{Level-\(i\) alignment using recycled pilot bits}
\label{alg:recycled-pilot-alignment}
\begin{algorithmic}[1]
\STATE \textbf{Input:} The level index~$i$, a fragment
\(\boldf=(\boldf[0],\ldots,\boldf[\ell-1])\) with
\(\ell\geq(1+\delta)t_i\log n\), the (already recovered) set \(\cS_i\), and the
bit values \(\{\boldc[x]\in\{0,1\}:x\in\cS_i\}\). 
\STATE \textbf{Output:} A unique starting position \(s_{\boldf}\), or an
ambiguity declaration.
\STATE Initialize \(\cM_i(\boldf)\gets\varnothing\).
\FOR{\(a=0\) to \(n-(1+\delta)t_i\log n\)}
    \STATE Set \(\textsc{consistent}\gets\textsc{true}\).
    \FOR{\(j=0\) to \((1+\delta)t_i\log n-1\)}
        \IF{\(a+j\in\cS_i\) and \(\boldf[j]\neq \boldc[a+j]\)}
            \STATE Set \(\textsc{consistent}\gets\textsc{false}\) and break.
        \ENDIF
    \ENDFOR
    \IF{\(\textsc{consistent}=\textsc{true}\)}
        \STATE Set
        \(\cM_i(\boldf)\gets\cM_i(\boldf)\cup\{a\}\).
    \ENDIF
\ENDFOR
\IF{\(\cM_i(\boldf)=\{s_{\boldf}\}\)}
    \STATE Place \(\boldf[j]\) at global position \(s_{\boldf}+j\) for every
    \(j\in\{0,1,\cdots,|\boldf|-1\}\), and return \(s_{\boldf}\).
\ELSE
    \STATE Declare \(\boldf\) ambiguous and discard it at level \(i\).
\ENDIF
\end{algorithmic}
\end{algorithm}

The recursion continues through levels \(i-1,i-2,\ldots,1\).
At the final level, the decoder processes every fragment of length at least
\((1+\delta)t_1\log n\), retaining those whose global starting positions are
uniquely determined. After level \(1\) has been decoded, the known set becomes
\(\cS_0=\cA_1\mathbin{\dot\cup}\cdots\mathbin{\dot\cup}\cA_J
\mathbin{\dot\cup}\cS_J=\{0,1,\ldots,n-1\}\). Thus, every level codeword
\(\boldu^{\,1},\ldots,\boldu^{\,J}\) has been recovered, and their
information bits can be combined to reconstruct the original message.
Section~\ref{Section:rate-and-reliability} analyzes the conditions by which the decoding process is guaranteed to succeed, and computes the resulting rate.


\section{Rate Calculation and Reliable Decoding}\label{Section:rate-and-reliability}
In this section, we provide  proofs and calculate the resulting rate of our code construction. 
The pilot-recycling technique from Section~\ref{Section:proposed-work encoding} is discussed in Section~\ref{subsection:coverage-and-ambiguity}, the required properties of random binary linear codes are discussed in Section~\ref{subsection:erasure-reliability}, and the resulting rate is calculated in Section~\ref{subsection:rate-calculation}.

\subsection{Pilot recycling ambiguity}\label{subsection:coverage-and-ambiguity}
In this subsection, we prove that the fragments discarded as ambiguous by
Algorithm~\ref{alg:recycled-pilot-alignment}, which compares a candidate
fragment in level~$i$ with bits in $\cS_i$, have total length $o(n)$. 
The fact that there is at least one consistent position in Algorithm~\ref{alg:recycled-pilot-alignment} is immediate since the torn-paper channel introduces no bit-level errors.
The proof has three steps: 
\begin{enumerate}
    \item For every~$i\in[J]$  and every level-\(i\) fragment of length at least
    \((1+\delta)t_i\log n\), we provide a lower bound on the number of bits in the fragment which originate from the already-recovered bits~$\cS_i$ (Lemma~\ref{lemma:usable-level-signature}).
    \item A fragment of length at least
    \((1+\delta)t_i\log n\) is ambiguous if and only if, in addition to its true global location, at least one more false global location also passes the consistency test against~$\cS_i$.
    We prove that this event has vanishing probability, and therefore the probability that a fragment is declared ambiguous is vanishing as well (Lemma~\ref{lemma:recycled-fragment-ambiguity}).
    \item Finally, we show that the event that a fragment is declared ambiguous occurs rarely enough, so that the \textit{total} length of all ambiguous fragments discarded by the decoder is \(o(n)\) (Lemma~\ref{lemma:vanishing-ambiguous-mass}).
\end{enumerate}


\begin{lemma}\label{lemma:usable-level-signature}
Fix $i\in[J-1]$, assume that levels $i+1,\ldots,J$ have been decoded correctly, i.e., every bit in $\cS_i$ is known to the decoder.
Then, for all sufficiently large $n$, every fragment of length at least
$(1+\delta)t_i\log n$ contains at least $(1+\delta/2)(t_i/T)\log n$ bits from the fixed De Bruijn sequence and at least $(1+\delta/2)(1-t_i/T)\log n$ bits from the already decoded level codewords $\bar{\bolds}^{\,i+1},\ldots,\bar{\bolds}^{\,J}$, where
$\bar{\bolds}^{\,j}\in\bar{\cC}_{\mathrm{er}}^j$ for all~$j\in\{i+1,\ldots,J\}$.
\end{lemma}
\begin{remark}\label{remark:special bits in S_j}
For the latter count $(1+\delta/2)(1-t_i/T)\log n$ in Lemma~\ref{lemma:usable-level-signature}, we only count the $\beta-1$ positions in each
length-$\beta$ RLL block that come from the shifted codeword in
$\cC_{\mathrm{sh}}^j$, $j\in\{i+1,\ldots,J\}$, and do not count the fixed $1$ inserted by the $\mathrm{RLL}(0,\beta-1)$ encoder.
Furthermore, it follows from Lemma~\ref{lemma:usable-level-signature} that every fragment of length at least $(1+\delta)t_i\log n$ contains at least $(1+\delta/2)\log n$ known bits of the two types"
\end{remark}

\begin{proof}[Proof of Lemma~\ref{lemma:usable-level-signature}]
We prove the statement using induction from level \(J\) down to level \(1\).

 \paragraph*{Base case at level \(J\)} 
This case has already been analyzed in Lemma~\ref{lemma:initial-debruijn-long-enough}.  
In particular, the proof of Lemma~\ref{lemma:initial-debruijn-long-enough} shows that, after removing the marker bits and the two \(1\)'s added in every length-\(\beta\) block by the
\(\widetilde{\mathrm{RLL}}(0,\beta-2)\) encoder, every fragment in
\(\cF_J\) contains at least
$(1+\delta)\log n-o(\log n)$ bits from the original De Bruijn sequence. 
Hence, for all sufficiently large \(n\), it contains at least
$(1+\delta/2)\log n=(1+\delta/2)\frac{t_J}{T}\log n$ such bits, which concludes the base case.  

\paragraph*{Inductive step}
Fix \(i\in[J-1]\), and assume that levels \(i+1,\ldots,J\) have been
decoded correctly.  
In particular, before level \(i+1\) was decoded, all
bits in \(\cS_{i+1}\) were already known, and decoding level \(i+1\)
recovers all bits in \(\cA_{i+1}\). 
By Definition~\ref{definition: Si and Ai},
$\cS_i=\cA_{i+1}\mathbin{\dot\cup}\cS_{i+1},$
and therefore every bit in \(\cS_i\) is known before the decoder begins level~\(i\);  
by recursion, it follows that
$\cS_i=\cA_{i+1}\mathbin{\dot\cup}\cA_{i+2}\mathbin{\dot\cup}\cdots\mathbin{\dot\cup}\cA_J\mathbin{\dot\cup}\cS_J$.
Thus, the known bits at level \(i\) consist of the fixed De Bruijn sequence bits in \(\cS_J\) together with the bits recovered from levels \(i+1,\ldots,J\).

We prove the inductive statement by considering bits from the fixed De Bruijn sequence bits and bits from the already decoded level codewords separately. 
By Definition~\ref{definition:modified-pilot}, one pilot bit is placed every \(T\) codeword positions.  
Hence, a fragment of length at least
\((1+\delta)t_i\log n\) contains at least $(1+\delta)\frac{t_i}{T}\log n$ bits from the pilot sequence.  
Similar to the analysis of Lemma~\ref{lemma:initial-debruijn-long-enough}, removing the marker bits and the two \(1\)'s added by the
\(\widetilde{\mathrm{RLL}}(0,\beta-2)\) encoder removes only
\(o(\log n)\) bits.  
Therefore, the number of bits from the original
De Bruijn sequence in a fragment of length \((1+\delta)t_i\log n\) is at least $(1+\delta)\frac{t_i}{T}\log n-o(\log n),$ which is at least $(1+\delta/2)\frac{t_i}{T}\log n$ for all sufficiently large \(n\).

Next, we count the bits from the already decoded level codewords.
By the induction hypothesis, the bits from levels \(i+1,\ldots,J\) codewords are already known. 
Hence, at level \(i\) we wish to count the number of bits that a fragment of length at least
\((1+\delta)t_i\log n\) contains in 
$\cA_{i+1}\mathbin{\dot\cup}\cA_{i+2}
\mathbin{\dot\cup}\cdots\mathbin{\dot\cup}\cA_J.$
By Lemma~\ref{lemma:periodic-coordinate-properties}, in every block of~\(Q\) consecutive bits, the level-\(j\) codeword occupies
$|\cB_j|=Q\left(\frac{1}{t_{j-1}}-\frac{1}{t_j}\right)$
bits. 
Therefore, the already decoded level codewords
\(i+1,\ldots,J\) together occupy
\begin{align*}
\sum_{j=i+1}^{J}|\cB_j|
=
Q\sum_{j=i+1}^{J}
\left(\frac{1}{t_{j-1}}-\frac{1}{t_j}\right) 
=
Q\left(\frac{1}{t_i}-\frac{1}{T}\right)
\label{eq:recycled-bits-per-Q-block}
\end{align*}
bits in every length-\(Q\) block (recall that~\(t_J=T\)).
Hence, since a fragment of
length \((1+\delta)t_i\log n\) contains at least
$\lceil(1+\delta)t_i\log n/Q\rceil-2$ complete blocks of length~$Q$ (omitting potentially incomplete blocks at the suffix and prefix), it follows that a fragment of
length \((1+\delta)t_i\log n\) contains at least  $$\left(\lceil(1+\delta)t_i\log n/Q\rceil-2\right)\cdot Q\left(\frac{1}{t_i}-\frac{1}{T}\right)\ge(1+\delta)t_i\log n\left(\frac{1}{t_i}-\frac{1}{T}\right)-2Q\left(\frac{1}{t_i}-\frac{1}{T}\right)$$ bits from  $\cA_{i+1}\mathbin{\dot\cup}\cA_{i+2}
\mathbin{\dot\cup}\cdots\mathbin{\dot\cup}\cA_J$.
The bits in $\cA_{i+1}\mathbin{\dot\cup}\cA_{i+2}
\mathbin{\dot\cup}\cdots\mathbin{\dot\cup}\cA_J$, however, contain~$1$'s added by the \(\mathrm{RLL}(0,\beta-1)\) encoder. 
Therefore, since every length-\(\beta\) block of a level codeword
contains one \(1\) added by the
\(\mathrm{RLL}(0,\beta-1)\) encoder, it follows that after omitting those \(\mathrm{RLL}(0,\beta-1)\) bits from $\cA_{i+1}\mathbin{\dot\cup}\cA_{i+2}
\mathbin{\dot\cup}\cdots\mathbin{\dot\cup}\cA_J$ we are left with
\[
\left(1-\frac{1}{\beta}\right)
\left[
(1+\delta)t_i\log n
\left(\frac{1}{t_i}-\frac{1}{T}\right)
\right]
-2Q\beta=
(1+\delta)
\left(1-\frac{t_i}{T}\right)\log n
-
\frac{(1+\delta)(1-t_i/T)\log n}{\beta}-2Q\beta
\]
bits from~$\bar{\bolds}^{\,i+1},\ldots,\bar{\bolds}^{\,J}$.
Since \(Q\) is fixed, \(\beta=\omega(1)\), and
\(\beta=o(\log n)\), it follows that \(\frac{(1+\delta)t_i\log n}{\beta}+2Q\beta=o(\log n)\).
Therefore, the number of bits contained in fragment of
length \((1+\delta)t_i\log n\) from the already decoded level
codewords is
$(1+\delta)\left(1-\frac{t_i}{T}\right)\log n-o(\log n),$
and hence, for all sufficiently large \(n\), it is at least
$(1+\delta/2)\left(1-\frac{t_i}{T}\right)\log n$.

Adding bits from the fixed De Bruijn sequence and the already decoded level codewords gives us $$(1+\delta/2)\frac{t_i}{T}\log n+(1+\delta/2)\left(1-\frac{t_i}{T}\right)\log n=(1+\delta/2)\log n,$$ which concludes the proof.
\end{proof}
Lemma~\ref{lemma:usable-level-signature} calculates the number of already known bits that a fragment of length $(1+\delta)t_i\log n$ contains.
Next, we show that the probability of such fragment being ambiguous, and hence discarded in Algorithm~\ref{alg:recycled-pilot-alignment}, is low.
In what follows, for~$i\in[J-1]$, we say that a placement of a fragment against a partially decoded codeword is $i$-\textit{consistent} (consistent, for short) if every fragment bit agrees with every aligned bit in~$\cS_i$. 

\begin{lemma}\label{lemma:recycled-fragment-ambiguity}
For every fixed \(i\in[J-1]\), there exists a constant \(\rho_i>0\) such that for every fragment $\boldf=(\boldf[0],\boldf[1],\ldots,\boldf[(1+\delta)t_i\log n-1])$ of length \((1+\delta)t_i\log n\), the probability that there exist two or more $i$-consistent placements 
is at most $n^{-\rho_i+o(1)}$.
In particular, the constant \(\rho_i\) may be chosen to satisfy
\[
0<\rho_i<
\min\left\{
\frac{\delta}{2},
\left(1+\frac{\delta}{2}\right)\left(1-\frac{t_i}{T}\right)
\right\}.
\]
\end{lemma}

\begin{proof}
Let the true starting position for a fragment~$\boldf$ of length $(1+\delta)t_i\log n$ be \(s\in\{0,1,\ldots,n-1\}.\)
By Lemma~\ref{lemma:usable-level-signature}, every candidate placement of
length \((1+\delta)t_i\log n\) contains at least
\(k_{p}=(1+\delta/2)(t_i/T)\log n\) fixed De Bruijn bits and
at least
\(k_{r}=(1+\delta/2)(1-t_i/T)\log n\) bits from $\bar{\bolds}^{\,i+1},\ldots,\bar{\bolds}^{\,J}$; see Remark~\ref{remark:special bits in S_j}.

Suppose that there exists at least one false placement \(a \neq s\) that is also consistent.
It suffices to bound the probability of the union of these events over all \(a \neq s\).
For a false placement to be consistent, every fragment bit mapped to one of the known positions (i.e.,~$\cS_i$ minus the RLL bits and marker bits, see Remark~\ref{remark:special bits in S_j}) must agree with the known bit at that position. 
We first observe the randomness of the level-codeword bits used in these
comparisons. Fix the message and all parity-check matrices. For every
level \(h\in[J]\), we have
\(\boldv^{\,h}=\boldu^{\,h}\oplus\boldz^{\,h}\), where
\(\boldz^{\,h}\) is uniform over
\(\mathbb F_2^{n_{\mathrm{er}}^h}\), and the shifts are independent
across different levels. Therefore, the coordinates of
\(\boldv^{\,1},\ldots,\boldv^{\,J}\) are mutually independent
\(\operatorname{Bernoulli}(1/2)\) random variables. Moreover, by
Remark~\ref{remark:special bits in S_j}, each of the~\(k_r\) positions
counted in Lemma~\ref{lemma:usable-level-signature} contains a bit from one of the shifted level codewords, rather than a marker bit or
a~\(1\) added by the RLL encoder. Such distinct codeword positions
correspond to distinct coordinates of the shifted level codewords.

Our proof strategy is as follows. 
For a false placement disjoint from the true placement, matching with $k_{p}$ fixed De Bruijn bits reduces the number of possible false positions, while matching with the $k_{r}$ bits from the already decoded level codewords reduces the probability of that placement.
For a false placement overlapping the true placement---of which there are at most \(3(1+\delta)t_i\log n\)---$k_{r}$ bits from the already decoded level codewords alone are used to prove that such false placement is unlikely to happen.

We first consider false placements disjoint from the true placement, namely
\( |a-s|\geq(1+\delta)t_i\log n \).
By Lemma~\ref{lemma:usable-level-signature}, the fragment contains at least
\(k_p\) bits from the fixed De Bruijn sequence.
The decoder does not need to know which bits of the fragment are from the
fixed De Bruijn sequence.
Rather, once a candidate \(a\) is chosen, the alignment determines which
fragment indices
\(\boldf[j]\), \(j\in\{0,1,\ldots,(1+\delta)t_i\log n-1\}\),
would be placed at the fixed De Bruijn sequence positions.
For the candidate to be consistent, all of these observed bits
\(\boldf[j]\) must agree with the corresponding known fixed De Bruijn bits.

Furthermore, since the De Bruijn sequence is of order
\(\log N_{q,n}\), any fixed binary sequence of length \(k_p\) occurs
in \(\boldq\) at most
\(2^{\max\{\log N_{q,n}-k_p,0\}}\) times.
Indeed, if \(k_p\leq\log N_{q,n}\), then the length-\(k_p\) sequence
has at most \(2^{\log N_{q,n}-k_p}\) possible extensions to a
length-\(\log N_{q,n}\) sequence, and each such extension occurs once
in the De Bruijn sequence. If \(k_p>\log N_{q,n}\), then the sequence
can occur at most once, since its first \(\log N_{q,n}\) bits already
determine its starting position uniquely. These two cases are expressed
together by the exponent
\(\max\{\log N_{q,n}-k_p,0\}\).

Now fix one of these remaining disjoint false placements and select
\(k_r\) of its comparisons against known bits from the already decoded
level codewords. Since the false placement is disjoint from the true
placement, the codeword positions used as the known bits in these
\(k_r\) comparisons lie outside the positions occupied by the fragment
under its true placement. Hence, by the one-to-one correspondence
between these codeword positions and coordinates of the shifted level
codewords, none of the corresponding random-shift coordinates occurs
among the observed fragment bits.

Condition on the message, all parity-check matrices, the channel cuts,
the entire fragment \(\boldf\), and all random-shift coordinates except
the \(k_r\) distinct coordinates corresponding to these known
level-codeword bits. Under this conditioning, whether the fixed
candidate passes the De Bruijn comparisons is already determined.
The \(k_r\) remaining level-codeword bits are still mutually independent
\(\operatorname{Bernoulli}(1/2)\) random variables, while the
corresponding fragment bits are fixed. Therefore,
\[
\Pr\left(
\text{all \(k_r\) fragment bits aligned against level-codeword bits are consistent}
\right)
\leq 2^{-k_r}.
\]
Since this bound holds after conditioning on the fragment and on the
outcome of the De Bruijn comparisons, it also holds without this
conditioning.

According to
Lemma~\ref{lemma:number-of-debruijn-subsequences-in-fragment} in Appendix~\ref{appendix:rate-reliability-estimates}, there are
at most \(n^{o(1)}\) candidate length-\(k_p\) De Bruijn subsequences in
a segment of length \((1+\delta)t_i\log n\).
Applying a union bound over all disjoint false placements that survive
the De Bruijn comparisons gives
\[
\Pr\left(
\text{there exists a consistent disjoint false placement}
\right)
\leq
n^{o(1)}
2^{\max\{\log N_{q,n}-k_p,0\}-k_r}.
\]
We now consider the two cases in the maximum.
On the one hand, if \(k_p\leq\log N_{q,n}\), then
\[
\max\{\log N_{q,n}-k_p,0\}-k_r
=
\log N_{q,n}-k_p-k_r.
\]
Since
\(k_p+k_r=(1+\delta/2)\log n\) by Lemma~\ref{lemma:usable-level-signature} and since \(\log N_{q,n}=\log n+O(1)\), it follows that
\[
\log N_{q,n}-k_p-k_r
\leq
-\frac{\delta}{2}\log n+O(1).
\]
Hence, in this case,
\[
\Pr\left(
\text{there exists a consistent disjoint false placement}
\right)
\leq
n^{-\frac{\delta}{2}+o(1)}.
\]
On the other hand, if \(k_p>\log N_{q,n}\), then the maximum equals \(0\), and therefore
\[
\Pr\left(
\text{there exists a consistent disjoint false placement}
\right)
\leq
n^{o(1)}2^{-k_r}
=
n^{-\left(1+\frac{\delta}{2}\right)
\left(1-\frac{t_i}{T}\right)+o(1)}.
\]
Thus, in both cases,
\begin{equation}\label{equation:disjoint false placement}
    \Pr\left(
\text{there exists a consistent disjoint false placement}
\right)
\leq
n^{-\min\left\{
\frac{\delta}{2},
\left(1+\frac{\delta}{2}\right)
\left(1-\frac{t_i}{T}\right)
\right\}+o(1)}.
\end{equation}

We next consider a fixed false placement \(a\neq s\) that overlaps the
true placement, and let \(d\triangleq a-s\). Each of the~\(k_r\)
comparisons against previously decoded level-codeword bits induces an
equality constraint
\[
\boldc[s+j]=\boldc[s+j+d]
\]
between two codeword coordinates whose indices differ by \(d\). By
Remark~\ref{remark:special bits in S_j}, the coordinate
\(\boldc[s+j+d]\) in each such comparison is an original bit from one
of the shifted level codewords, rather than a marker bit or a \(1\)
added by the RLL encoder.
Construct a graph whose vertices are the codeword coordinates involved
in these comparisons and whose edges represent these equality
constraints, with each edge directed from \(s+j\) to \(s+j+d\).
Since every edge has displacement \(d\), each connected component lies
within a single residue class modulo \(|d|\) and is a subgraph of a
path. Hence, the comparison graph is a forest, or more specifically, a
collection of paths.

Consider one connected component containing \(k\) edges. With the above
orientation, every vertex in the component except for one endpoint is
the head of exactly one edge. Therefore, these \(k\) vertices are all
positions containing original shifted level-codeword bits. Moreover,
they correspond to \(k\) distinct coordinates of the random shifts and
are therefore mutually independent
\(\operatorname{Bernoulli}(1/2)\) random variables. The remaining
endpoint may be either a shifted level-codeword bit or a deterministic
pilot or an RLL bit, which does not affect the argument.
Conditioning on the value of this remaining endpoint, the \(k\)
equality constraints along the path successively require the \(k\)
independent \(\operatorname{Bernoulli}(1/2)\) bits to take prescribed
values. Hence, all \(k\) constraints in this component are satisfied
with probability at most \(2^{-k}\). Since different connected
components involve distinct shifted-codeword coordinates, multiplying
over all components gives
\[
\Pr\bigl(
\text{a fixed overlapping false placement is consistent}
\bigr)
\leq
2^{-k_r}.
\]

A union bound over the \(O(\log n)\) possible overlapping false placements gives
\begin{equation}\label{equation:overlap false placement}
    \Pr\left(
\text{there exists a consistent overlapping false placement}
\right)\leq
3(1+\delta)t_i\log n\,2^{-k_{r}}=
n^{-
\left(1+\frac{\delta}{2}\right)
\left(1-\frac{t_i}{T}\right)+o(1)}.
\end{equation}

By adding~\eqref{equation:disjoint false placement} and~\eqref{equation:overlap false placement} we obtain the following via union bound:
\[
\Pr\bigl(\text{there exists at least two $i$-consistent placements})
\leq
n^{-\rho_i+o(1)}
\]
for every constant
\[
0<\rho_i<
\min\left\{
\frac{\delta}{2},
\left(1+\frac{\delta}{2}\right)
\left(1-\frac{t_i}{T}\right)
\right\}.\qedhere
\]
\end{proof}

The previous lemma bounds the ambiguity probability for one possible
starting position of a fragment.  
The next lemma converts this bound into the total length of all ambiguous fragments; this quantity is required for the erasure analysis
of Section~\ref{subsection:erasure-reliability}.

\begin{lemma}\label{lemma:vanishing-ambiguous-mass}
Let $A_{i,n}^{\star}$ be the total length of all fragments of length
at least $(1+\delta)t_i\log n$ that would be declared ambiguous if the
alignment test used the true underlying bits on $\cS_i$. 
Then, for every $\varepsilon > 0$ we have that
$$
\lim_{n\to\infty} \Pr\left(\frac{A_{i,n}^{\star}}{n} > \varepsilon\right) = 0,
$$
where the probability is taken over the random shift $\boldz^i$ in Definition~\ref{definition:random-linear-erasure-codes}.
\end{lemma}

\begin{proof}
For every possible starting position
\(s\in\{0,1,\ldots,n-1\}\), let \(B_s\in\{0,1\}\) be a random variable which indicates that the first
\((1+\delta)t_i\log n\) codeword bits beginning at \(s\) admit at least
one additional consistent placement when the true codeword bits on
\(\cS_i\) are used in the alignment test.  If fewer than
\((1+\delta)t_i\log n\) codeword bits remain after position \(s\), we set
\(B_s=0\).  
Then, it follows from Lemma~\ref{lemma:recycled-fragment-ambiguity} that
\[
\mathbb E_{\mathrm{shifts}}[B_s]
=
\Pr(B_s=1)
\leq
n^{-\rho_i+o(1)}
\]
for every~$s$.

The level-\(i\) alignment test uses only the first
\((1+\delta)t_i\log n\) bits of every sufficiently long fragment.
Consequently, if a fragment beginning at \(s\) has length at least
\((1+\delta)t_i\log n\) and is ambiguous under the consistency
test, then these first \((1+\delta)t_i\log n\) bits admit at least one
additional consistent placement, and hence \(B_s=1\).

Recall from Section~\ref{section:torn-paper-channel} that the torn-paper
channel divides the codeword into \(K\) fragments, and that \(N_\ell\) is the
length of the \(\ell\)-th fragment, where the final fragment is truncated
at the end of the codeword.
Further, let \(S_\ell\) be the starting position of~$N_\ell$.
More precisely, \(S_1=0\) and
\(S_\ell=\sum_{j=1}^{\ell-1}N_j\) for every
\(\ell\in\{2,\ldots,K\}\).
Since the fragments form a disjoint partition
of the entire length-\(n\) codeword, their lengths satisfy
$
\sum_{\ell=1}^{K}N_\ell=n.
$

If the \(\ell\)-th fragment is included in \(A_{i,n}^{\star}\), then it
has length at least \((1+\delta)t_i\log n\) and it is ambiguous under the
consistency test, and therefore \(B_{S_\ell}=1\).  Its contribution
to \(A_{i,n}^{\star}\) is its entire length \(N_\ell\), and hence
\[
A_{i,n}^{\star}
\leq
\sum_{\ell=1}^{K}B_{S_\ell}N_\ell.
\]

We now explain how the expectation of~$A_{i,n}^{\star}$ is calculated.  
First fix all cut locations made by the torn-paper channel.  Once these cut locations are fixed, the
number of fragments \(K\), their starting positions
\(S_1,\ldots,S_K\), and their lengths \(N_1,\ldots,N_K\) are fixed.
The variables \(B_{S_1},\ldots,B_{S_K}\) remain random since they depend
on the random shifts~$\boldz_i$.  Since the random shifts  are independent of the
channel cuts, fixing the cuts does not change the uniform bound
\(\mathbb E_{\mathrm{shifts}}[B_{S_\ell}]
\leq n^{-\rho_i+o(1)}\).

We therefore first take expectation over the random shifts  for fixed cut
locations and then average over the random channel cuts.  By the law of
total expectation and linearity of expectation,
\[
\begin{aligned}
\mathbb E[A_{i,n}^{\star}]
&=
\mathbb E_{\mathrm{cuts}}\!\left[
\mathbb E_{\mathrm{shifts}}\!\left[
A_{i,n}^{\star}
\,\middle|\,
\text{fixed cut locations}
\right]
\right]\\
&\leq
\mathbb E_{\mathrm{cuts}}\!\left[
\mathbb E_{\mathrm{shifts}}\!\left[
\sum_{\ell=1}^{K}B_{S_\ell}N_\ell
\,\middle|\,
\text{fixed cut locations}
\right]
\right]\\
&\overset{(\dagger)}{=}
\mathbb E_{\mathrm{cuts}}\!\left[
\sum_{\ell=1}^{K}
N_\ell\,
\mathbb E_{\mathrm{shifts}}[B_{S_\ell}]
\right]\\
&\leq
\mathbb E_{\mathrm{cuts}}\!\left[
n^{-\rho_i+o(1)}
\sum_{\ell=1}^{K}N_\ell
\right]=
n^{1-\rho_i+o(1)}
=
o(n),
\end{aligned}
\]
where in  $(\dagger)$ the fragment lengths \(N_\ell\) are taken outside the
expectation over the random shifts  since they are already fixed once the
channel cuts are fixed.  The last equality follows since
\(\sum_{\ell=1}^{K}N_\ell=n\) and since \(\rho_i>0\).
Finally, for every fixed \(\varepsilon>0\), Markov's inequality gives
\[
\Pr\left(
\frac{A_{i,n}^{\star}}{n}>\varepsilon
\right)
\leq
\frac{\mathbb E[A_{i,n}^{\star}]}{\varepsilon n}
\leq
\frac{n^{-\rho_i+o(1)}}{\varepsilon}
\overset{n\to\infty}{\longrightarrow} 0.\qedhere
\]
\end{proof}

\subsection{Induced erasure patterns and random linear decoding}
\label{subsection:erasure-reliability}
We now complete the decoding-reliability analysis by showing that the partially observed level codewords obtained in Section~\ref{subsection:level-erasure-decoding} contain sufficiently few erasures for the corresponding random binary linear codes to recover all missing bits. 
Fix a level \(i\in[J]\).  By
Lemma~\ref{lemma:level-rll-correct}, the output
\(\widehat{\boldu}^{\,i}\in\{0,1,\star\}^{n_{\mathrm{er}}^i}\) of
Algorithm~\ref{alg:level-rll} is a partially observed codeword of
\(\cC_{\mathrm{er}}^i\).  
Let
\(\cE_i\subseteq\{0,1,\ldots,n_{\mathrm{er}}^i-1\}\) denote the positions
of its erased bits. 
 We recover these bits using the standard
erasure-decoding procedure for binary linear codes.  Since
\(\cC_{\mathrm{er}}^i=\ker(H_i)\), where
\(H_i\in\mathbb F_2^{r_i\times n_{\mathrm{er}}^i}\) is the parity-check
matrix from Definition~\ref{definition:random-linear-erasure-codes},
decoding reduces to solving a linear system over \(\mathbb F_2\) for the
erased bits using Gaussian elimination.
Let \((H_i)_{\cE_i}\) denote the submatrix of \(H_i\) formed by the
columns indexed by the erased-coordinate set \(\cE_i\).  The erased bits
in \(\widehat{\boldu}^{\,i}\) are recovered uniquely if and only if these
columns are linearly independent, or equivalently, if
\(\rank((H_i)_{\cE_i})=|\cE_i|\).  The probability that a uniformly
random parity-check matrix satisfies this condition for a fixed erasure
pattern is analyzed in
Lemma~\ref{lemma:random linear code random erasures}, given shortly.

Even though the torn paper channel is inherently memoryless, in the sense that different cut locations are independent, the erasure channel that is induced by our code construction is \emph{not} memoryless.
Since the torn paper channel cuts the codeword into pieces, the discarded fragments typically
induce bursts of consecutive erasures. 
Consequently, standard coding tools and concentration bounds for the
memoryless binary erasure channel (BEC) cannot be applied directly. 
We therefore do not model these erasures as independent binary erasures. Instead, following the approach used for the random linear code analysis in~\cite{liu2026improved}, we first prove that the induced erasure pattern belongs with high probability to a family of \textit{good erasure patterns at level $i$}, and then prove that a uniformly random parity-check matrix corrects every fixed pattern in this family with high probability.

For a fixed $\eta>0$ and each level $i\in[J]$, define the
\textit{set of good erasure patterns at level $i$}, denoted by $\cT_i$,
as the collection of all subsets of $[n_{\mathrm{er}}^i]$ whose size
is at most $(1-R_{\mathrm{er}}^i-\eta)n_{\mathrm{er}}^i$; that is,
\begin{align}\label{equation: definition of mathcal T}
     \cT_i
     = \binom{[n_{\mathrm{er}}^i]}
     {\leq (1-R_{\mathrm{er}}^i-\eta)n_{\mathrm{er}}^i}.
\end{align}
The following known lemma relies on independence properties of random binary vectors.
\begin{lemma}\label{lemma:random linear code random erasures}
Fix a level $i\in[J]$. Let $\cC_{\mathrm{er}}^i$ be the random binary
linear code from Definition~\ref{definition:random-linear-erasure-codes},
with parity-check matrix
$H_i\in\mathbb{F}_2^{r_i\times n_{\mathrm{er}}^i}$, where
$r_i=(1-R_{\mathrm{er}}^i)n_{\mathrm{er}}^i$.
For every fixed erasure pattern $\cE_i\in\cT_i$, we have
$
\Pr_{H_i}\bigl(\rank((H_i)_{\cE_i})\neq|\cE_i|\bigr)
\leq 2^{-\eta n_{\mathrm{er}}^i}.
$
\end{lemma}
\begin{proof}
Fix $\cE_i\in\cT_i$ and consider it as a deterministic erasure pattern.
Let $e_i=|\cE_i|$. The code $\cC_{\mathrm{er}}^i$ successfully decodes
$\cE_i$ if and only if the columns of $(H_i)_{\cE_i}$ are linearly
independent, or equivalently, if $\rank((H_i)_{\cE_i})=e_i$.
For any integer~$1\le a\le r_i$, conditioned on the first $a$ columns being linearly independent, 
the probability that the next random column belongs to their span is $2^{a-r_i}$.
Hence, a union bound gives
$
\Pr_{H_i}\bigl(\rank((H_i)_{\cE_i})\neq e_i\bigr)
\leq \sum_{a=0}^{e_i-1}2^{a-r_i}<2^{e_i-r_i}.
$
Since $\cE_i\in\cT_i$, we have
$e_i\leq(1-R_{\mathrm{er}}^i-\eta)n_{\mathrm{er}}^i$, whereas
$r_i=(1-R_{\mathrm{er}}^i)n_{\mathrm{er}}^i$. Therefore
$e_i-r_i\leq-\eta n_{\mathrm{er}}^i$, which proves the stated bound.
For every fixed level $i$, $n_{\mathrm{er}}^i$ grows linearly with $n$;
therefore, the failure probability tends to zero exponentially as
$n\to\infty$.
\end{proof}
For the following lemma, recall that $\cE_i$ is the
random set of erased coordinates in the partially observed level-$i$
codeword $\widehat{\boldu}^{\,i}$ returned by
Algorithm~\ref{alg:level-rll} for each level $i\in[J]$, and
recall that $R_{\mathrm{er}}^i$ and
$n_{\mathrm{er}}^i$ are the rate and block length of the
random binary linear code $\cC_{\mathrm{er}}^i$ in
Definition~\ref{definition:random-linear-erasure-codes}, respectively. 
\begin{lemma}\label{lemma:bound on number of erasures}
Consider the decoding procedure after fragment alignment and
inverse RLL decoding. 
For a fixed
message, let $\mathsf E_{>i}$ be the event that levels
$i+1,\ldots,J$ are decoded correctly.
Then, for every $i\in[J]$,
$
\Pr\bigl(\mathsf E_{>i}\cap\{\cE_i\notin\cT_i\}\bigr)
\xrightarrow[n\to\infty]{}0.
$
\end{lemma}

\begin{proof}
Let $\cG_i^c \subseteq \{0,1,\ldots,n-1\}$ denote the set of coordinates that are unavailable to the decoder after the level-$i$ alignment procedure.
This set consists of all fragments of length strictly less than $(1+\delta)t_i\log n$, which are not considered at level~$i$, together with the fragments of length at least $(1+\delta)t_i\log n$ that undergo the consistency test and are declared ambiguous, and hence discarded at level~$i$. 

For a constant \(\gamma\), the \textit{coverage} \(V_\gamma\)~\cite[Def.~1]{shomorony2021torn} is the fraction of output bits that reside in fragments of length at least \(\gamma\log n\), i.e.,
\begin{align}\label{Definition: coverage of fragment}
    V_\gamma=\frac{1}{n}\sum_{j=1}^{K}N_j\mathbf{1}\!\left\{N_j\geq\gamma\log n\right\},
\end{align}
where \(\mathbf{1}\!\left\{N_j\geq\gamma\log n\right\}\) is the indicator of the event \(N_j\geq\gamma\log n\). 
In particular, \(V_{(1+\delta)t_i}\) denotes the fraction of codeword coordinates contained in fragments of length at least \((1+\delta)t_i\log n\). 
Then, we have that
$$
|\cG_i^c| \leq n(1 - V_{(1+\delta)t_i}) + A_{i,n}^{\star},
$$
where~\(A_{i,n}^{\star}\) is the total length of the ambiguous fragments under the consistency test, as in Lemma~\ref{lemma:vanishing-ambiguous-mass}.

Recall from Lemma~\ref{lemma:periodic-coordinate-properties} and Definition~\ref{definition:random-linear-erasure-codes} that $|\cA_i|=\frac{n}{Q}|\cB_i|$ and $n_{\mathrm{er}}^i=\frac{\beta-1}{\beta}|\cA_i|,$ and thus $\frac{(\beta-1)|\cB_i|}{Q\beta}=\frac{n_{\mathrm{er}}^i}{n}.$
Now consider one discarded fragment, which may either have length less than
\((1+\delta)t_i\log n\) or have length at least this threshold but be
declared ambiguous.  
Recall that, within every \(Q\beta\) consecutive
codeword positions, each residue \(b\in\cB_i\) appears exactly \(\beta\)
times.  
Among these \(\beta\) positions,
exactly \(\beta-1\) contain bits from
\(\boldv^{\,i}=\boldu^{\,i}\oplus\boldz^{\,i}\) while the remaining
position contains the \(1\) added by the
\(\mathrm{RLL}(0,\beta-1)\) encoder.  
Hence, every complete block of
\(Q\beta\) consecutive codeword positions contains exactly
\((\beta-1)|\cB_i|\) level-\(i\) bits.  Since
\[
\frac{(\beta-1)|\cB_i|}{Q\beta}
=
\frac{n_{\mathrm{er}}^i}{n},
\]
the fraction of level-\(i\) bits in every complete length-\(Q\beta\)
block is exactly \(n_{\mathrm{er}}^i/n\).

For a discarded fragment, all complete length-\(Q\beta\) blocks contained
in the fragment therefore contribute exactly this fraction of level-\(i\)
bits.  At most one incomplete block can occur at the beginning of the
fragment and at most one incomplete block can occur at the end.  Together,
these two incomplete blocks contain fewer than \(2Q\beta\) codeword
bits.  
Since \(\cG_i^c\) is the union of at most \(K\) discarded fragments,
summing this bound over all discarded fragments gives,
\[
|\cE_i|
\leq
\frac{n_{\mathrm{er}}^i}{n}|\cG_i^c|
+
2Q\beta K.
\]
Dividing by \(n_{\mathrm{er}}^i\), we obtain
\begin{equation}
\frac{|\cE_i|}{n_{\mathrm{er}}^i}
\leq
1-V_{(1+\delta)t_i}
+
\frac{A_{i,n}^{\star}}{n}
+
\frac{2Q\beta K}{n_{\mathrm{er}}^i}.
\label{eq:level-i-erased-fraction-bound}
\end{equation}

We first bound the final term in
\eqref{eq:level-i-erased-fraction-bound}.  By
Lemma~\ref{lemma:number-of-fragments} in Appendix~\ref{appendix:rate-reliability-estimates}, there exists a constant \(c>0\)
such that
\[
\Pr\left(
K>\frac{cn}{\log n}
\right)
\longrightarrow0.
\]
Whenever \(K\leq cn/\log n\), Definition~\ref{definition:random-linear-erasure-codes}
and Lemma~\ref{lemma:periodic-coordinate-properties} give
\[
\begin{aligned}
\frac{2Q\beta K}{n_{\mathrm{er}}^i}
\leq
\frac{2Q\beta(cn/\log n)}
{\frac{\beta-1}{\beta}\frac{n}{Q}|\cB_i|}
=
\frac{2cQ^2\beta^2}
{(\beta-1)|\cB_i|\log n}.
\end{aligned}
\]
Since \(Q\), \(c\), and \(|\cB_i|\) are fixed constants and
\(\beta=o(\log n)\), it follows that
$\frac{2Q\beta K}{n_{\mathrm{er}}^i}=o(1)$ for sufficiently large $n$, i.e. $\Pr(\frac{2Q\beta K}{n_{\mathrm{er}}^i}<\eta)\xrightarrow[]{}1$ for any fixed $\eta$ and for sufficiently large $n$.

We next consider the first two terms in \eqref{eq:level-i-erased-fraction-bound}.
Lemma~\ref{lemma:long-fragment-coverage} in Appendix~\ref{appendix:rate-reliability-estimates} shows that the coverage fraction
\(V_{(1+\delta)t_i}\) converges in probability to
$
\bigl(\alpha(1+\delta)t_i+1\bigr)
e^{-\alpha(1+\delta)t_i}
$ as~$n\to\infty$.
Therefore, for any fixed \(\eta>0\), with probability which tends to one as $n$ tends to infinity,
\[
V_{(1+\delta)t_i}
\geq
\bigl(\alpha(1+\delta)t_i+1\bigr)
e^{-\alpha(1+\delta)t_i}
-\frac{\eta}{2}.
\]
Moreover, Lemma~\ref{lemma:vanishing-ambiguous-mass} shows that, when
levels \(i+1,\ldots,J\) are decoded correctly, with probability which tends to one as $n$ tends to infinity we have that
$
\frac{A_{i,n}^{\star}}{n}
\leq
\frac{\eta}{2}.
$
Whenever these two inequalities hold simultaneously, we have
\[
\begin{aligned}
1-V_{(1+\delta)t_i}
+\frac{A_{i,n}^{\star}}{n}
&\leq
1-
\left[
\bigl(\alpha(1+\delta)t_i+1\bigr)
e^{-\alpha(1+\delta)t_i}
-\frac{\eta}{2}
\right]
+\frac{\eta}{2}\\
&=
1-
\bigl(\alpha(1+\delta)t_i+1\bigr)
e^{-\alpha(1+\delta)t_i}
+\eta.
\end{aligned}
\]
Therefore, we have shown that the first two terms of~\eqref{eq:level-i-erased-fraction-bound} on the right-hand side are at most
\[
1-
\bigl(\alpha(1+\delta)t_i+1\bigr)
e^{-\alpha(1+\delta)t_i}
+\eta
\]
with probability tending to one, while the final term is at most
\(\eta\) with probability which tends to one as $n\to\infty$.  Therefore, if
\[
\frac{|\cE_i|}{n_{\mathrm{er}}^i}
>
1-
\bigl(\alpha(1+\delta)t_i+1\bigr)
e^{-\alpha(1+\delta)t_i}
+2\eta,
\]
then at least one of these two bounds must fail.  By the union bound, the
probability of this event tends to zero.  Hence,
\[
\Pr\left(
\mathsf E_{>i}
\cap
\left\{
\frac{|\cE_i|}{n_{\mathrm{er}}^i}
>
1-
\bigl(\alpha(1+\delta)t_i+1\bigr)
e^{-\alpha(1+\delta)t_i}
+2\eta
\right\}
\right)
\longrightarrow0.
\]

Now, by Definition~\ref{definition:random-linear-erasure-codes} we have that
\begin{align*}
    R_{\mathrm{er}}^i&=
\bigl(\alpha(1+\delta)t_i+1\bigr)e^{-\alpha(1+\delta)t_i}-3\eta,\\
1-R_{\mathrm{er}}^i-\eta
&=
1-
\bigl(\alpha(1+\delta)t_i+1\bigr)
e^{-\alpha(1+\delta)t_i}
+
2\eta,
\end{align*}
and hence,
\[
\Pr\left(
\mathsf E_{>i}
\cap
\left\{
|\cE_i|>
(1-R_{\mathrm{er}}^i-\eta)n_{\mathrm{er}}^i
\right\}
\right)
\longrightarrow0.
\]
By the definition of \(\cT_i\) in
\eqref{equation: definition of mathcal T}, this is equivalent to
\[
\Pr\left(
\mathsf E_{>i}
\cap
\{\cE_i\notin\cT_i\}
\right)
\longrightarrow0,
\]
which concludes the proof.
\end{proof}
To show that the random binary linear code successfully corrects the
erasures at every level, we analyze the probability that all deeper
levels have been decoded correctly but level \(i\) fails. For a fixed
parity-check matrix \(H_i\), let
\begin{equation}\label{equation:P_e^{(i)})}
    P_e^{(i)}(H_i)\triangleq\text{Pr}(\mathsf E_{>i}\cap\{\text{level \(i\) fails to decode correctly}\}),
\end{equation}

where the probability is taken over the uniformly chosen message, the
torn-paper channel, the random shifts , and all parity-check matrices other
than \(H_i\).
We evaluate
\(\mathbb E_{H_i}[P_e^{(i)}(H_i)]\), where \(H_i\) is chosen uniformly
at random.
The main idea is to divide the level-\(i\) decoding failure according to
whether its erased-coordinate set \(\cE_i\) belongs to the good family
\(\cT_i\).
When \(\cE_i\in\cT_i\), Lemma~\ref{lemma:random linear code random erasures}
shows that a uniformly random \(H_i\) fails to correct this fixed pattern
with probability at most \(2^{-\eta n_{\mathrm{er}}^i}\). On the other
hand, the preceding good-erasure-pattern analysis shows that
\(\mathsf E_{>i}\) occurs while \(\cE_i\notin\cT_i\) with probability
tending to zero by Lemma~\ref{lemma:bound on number of erasures}. 
Combining these two bounds yields the following result.

\begin{lemma}\label{lemma:shift-independence}
Fix a level \(i\in[J]\) and a message. Let
\(\boldu^{\,i}=\boldu^{\,i}(H_i)\in\cC_{\mathrm{er}}^i\) be the
level-\(i\) codeword selected by the message, and recall that
\(\boldv^{\,i}=\boldu^{\,i}\oplus\boldz^{\,i}\), where
\(\boldz^{\,i}\) is uniform over
\(\mathbb F_2^{n_{\mathrm{er}}^i}\) and independent of \(H_i\).
Then \(\boldv^{\,i}\) is uniform over
\(\mathbb F_2^{n_{\mathrm{er}}^i}\) and is independent of \(H_i\).
Moreover, before level-\(i\) erasure decoding, the pair
\((\mathsf E_{>i},\cE_i)\) depends on \(H_i\) only through
\(\boldv^{\,i}\). Therefore,
\((\mathsf E_{>i},\cE_i)\) is independent of \(H_i\).
\end{lemma}

\begin{proof}
Fix the message and condition on the channel cuts, as well as on all parity-check
matrices \(H_j\) with \(j\neq i\), and on all random shifts \(\boldz^{\,j}\) with \(j\neq i\). 
For every fixed \(H_i\), the vector
\(\boldu^{\,i}(H_i)\) is fixed. Since \(\boldz^{\,i}\) is uniform and
independent of \(H_i\), the vector
\(\boldv^{\,i}=\boldu^{\,i}(H_i)\oplus\boldz^{\,i}\) is uniform over
\(\mathbb F_2^{n_{\mathrm{er}}^i}\), regardless of the value of \(H_i\).
Hence, \(\boldv^{\,i}\) is independent of \(H_i\).

For the ``moreover'' part, before level-\(i\) erasure decoding,
the matrix \(H_i\) is not used directly by the decoder. Its only effect
on the codeword is through the level-\(i\) codeword
\(\boldu^{\,i}(H_i)\), which, after applying the shift~$\boldz^i$, appears only
through \(\boldv^{\,i}\). 
The RLL encoding and placement of these bits
are deterministic with respect to \(\boldv^{\,i}\).
Therefore, under this conditioning, both \(\mathsf E_{>i}\) and
\(\cE_i\) depend on \(H_i\) only through \(\boldv^{\,i}\).
Since \(\boldv^{\,i}\) is independent of \(H_i\), the pair
\((\mathsf E_{>i},\cE_i)\) is independent of \(H_i\).
Removing the conditioning preserves this independence.
\end{proof}

\begin{lemma}\label{lemma:random_matrix_capacity}
For every fixed level \(i\in[J]\) we have $\lim_{n\to\infty}\mathbb E_{H_i}\!P_e^{(i)}(H_i)=0$.
\end{lemma}

\begin{proof}
Suppose that \(\mathsf E_{>i}\) occurs, meaning that levels
\(i+1,\ldots,J\) have been decoded correctly, and the decoder has the
correct codeword bits on \(\cS_i\). Hence the fragments retained at
level \(i\) are placed correctly, and
Algorithm~\ref{alg:level-rll} returns a partially observed version of the
correct codeword \(\boldu^{\,i}\in\cC_{\mathrm{er}}^i\), with erased
coordinates \(\cE_i\).
Therefore, level~\(i\) can fail only in one of two cases. Either
\(\cE_i\notin\cT_i\), or \(\cE_i\in\cT_i\) but the columns of
\((H_i)_{\cE_i}\) are linearly dependent.

By Lemma~\ref{lemma:shift-independence}, \(\mathsf E_{>i}\) and \(\cE_i\) do not depend on \(H_i\).
Consequently, for every fixed \(\cE\in\cT_i\), we may apply the fixed-erasure-pattern bound of
Lemma~\ref{lemma:random linear code random erasures}.
Hence, it follows that
\[
\begin{aligned}
\mathbb E_{H_i}\!\left[P_e^{(i)}(H_i)\right]
&\leq
\Pr\left(
\mathsf E_{>i}
\cap
\{\cE_i\notin\cT_i\}
\right)+
\sum_{\cE\in\cT_i}
\Pr\left(
\mathsf E_{>i}
\cap
\{\cE_i=\cE\}
\right)
\Pr_{H_i}\left(
\rank((H_i)_{\cE})\neq|\cE|
\right).
\end{aligned}
\]
Furthermore, $\Pr_{H_i}\left(
\rank((H_i)_{\cE})\neq|\cE|
\right)
\leq
2^{-\eta n_{\mathrm{er}}^i}$ for every fixed \(\cE\in\cT_i\) by Lemma~\ref{lemma:random linear code random erasures}, and therefore,
\[
\begin{aligned}
\mathbb E_{H_i}\!\left[P_e^{(i)}(H_i)\right]
&\leq
\Pr\left(
\mathsf E_{>i}
\cap
\{\cE_i\notin\cT_i\}
\right)+
2^{-\eta n_{\mathrm{er}}^i}
\sum_{\cE\in\cT_i}
\Pr\left(
\mathsf E_{>i}
\cap
\{\cE_i=\cE\}
\right)\\
&\leq
\Pr\left(
\mathsf E_{>i}
\cap
\{\cE_i\notin\cT_i\}
\right)
+
2^{-\eta n_{\mathrm{er}}^i}.
\end{aligned}
\]
The first term tends to zero as $n\to\infty$ by Lemma~\ref{lemma:bound on number of erasures}. Since \(\eta>0\) is fixed and \(n_{\mathrm{er}}^i\) grows
linearly with \(n\), the second term tends to zero exponentially, and consequently, $\lim_{n\to\infty}\mathbb E_{H_i}\!\left[P_e^{(i)}(H_i)\right]=0$.
\end{proof}
Having established vanishing decoding error at every fixed level, we
now account for all \(J\) levels simultaneously.  Recall that the decoder
processes the levels successively in the order
\(J,J-1,\ldots,1\).  Therefore, if the decoder fails, there
must be a first level \(i\) that fails, while all previous levels
\(i+1,\ldots,J\) have been decoded correctly.  We first use this
observation to show that the expected total decoding error of the random
multilevel construction vanishes, and then apply Markov's inequality to
the total error probability.
Let
\[
\boldH\triangleq(H_1,\ldots,H_J)
\qquad\text{and}\qquad
\boldz\triangleq
(\boldz^{\,1},\ldots,\boldz^{\,J})
\]
denote the collections of parity-check matrices and random shifts,
respectively.  For fixed \(\boldH\) and \(\boldz\), let
\(P_e^{\mathrm{total}}(\boldH,\boldz)\) denote the decoding error
probability of the complete multilevel decoder, averaged over a uniformly
chosen message and the torn-paper channel.

\begin{corollary}\label{corollary:total_decoding_error}
The expected total decoding error probability of the random multilevel
construction vanishes; that is,
\[
\lim_{n\to\infty}
\mathbb E_{\boldH,\boldz}
\left[
P_e^{\mathrm{total}}(\boldH,\boldz)
\right]
=
0.
\]
\end{corollary}

\begin{proof}
Take probability jointly over the uniformly chosen message, the torn-paper channel, the random shifts, and the random parity-check matrices.
If the decoder fails, then there exists a first
level \(i\in[J]\) at which decoding fails.  Since the decoder proceeds in
the order \(J,J-1,\ldots,1\), all levels \(i+1,\ldots,J\) must have
been decoded correctly before this first failure occurs.  Therefore,
\[
\{\text{overall decoding failure}\}
\subseteq
\bigcup_{i=1}^{J}
\left(
\mathsf E_{>i}
\cap
\{\text{level \(i\) fails to decode}\}
\right).
\]
By the union bound,
\[
\Pr(\text{overall decoding failure})
\leq
\sum_{i=1}^{J}
\Pr\left(
\mathsf E_{>i}
\cap
\{\text{level \(i\) fails to decode}\}
\right).
\]
By the definition of \(P_e^{(i)}(H_i)\) in~\eqref{equation:P_e^{(i)})},
\[
\Pr\left(
\mathsf E_{>i}
\cap
\{\text{level \(i\) fails to decode}\}
\right)
=
\mathbb E_{H_i}\!\left[P_e^{(i)}(H_i)\right].
\]
Hence,
\begin{align}\label{equation:ExpectationFinal}
\mathbb E_{\boldH,\boldz}
\left[
P_e^{\mathrm{total}}(\boldH,\boldz)
\right]
\leq
\sum_{i=1}^{J}
\mathbb E_{H_i}\!\left[P_e^{(i)}(H_i)\right].
\end{align}
Lemma~\ref{lemma:random_matrix_capacity} shows that $\mathbb E_{H_i}\!P_e^{(i)}(H_i)\to0$ as~$n\to\infty$ for every fixed
\(i\in[J]\), and hence the r.h.s of~\eqref{equation:ExpectationFinal} tends to zero as a finite sum, which concludes the proof.
\end{proof}

The following lemma applies Markov's inequality to the total decoding
error probability and shows that a randomly chosen multilevel code design
achieves vanishing decoding error with high probability.

\begin{lemma}\label{lemma:high_probability_existence}
There exists a nonnegative sequence $a_n$, which tends to zero as~$n$ tends to infinity, such that
\[
\Pr_{\boldH,\boldz}\left(
P_e^{\mathrm{total}}(\boldH,\boldz)\leq a_n
\right)
\xrightarrow[n\to\infty]{}1.
\]
\end{lemma}

\begin{proof}
Let
\[
\epsilon_n
\triangleq
\mathbb E_{\boldH,\boldz}
\left[
P_e^{\mathrm{total}}(\boldH,\boldz)
\right].
\]
By Corollary~\ref{corollary:total_decoding_error},
\(\epsilon_n\to0\) as $n\to\infty$. Set
\[
a_n\triangleq\sqrt{\epsilon_n}+\frac{1}{n}.
\]
Then \(a_n>0\) and \(a_n\to0\) as $n\to\infty$. Markov's inequality gives
\[
\Pr_{\boldH,\boldz}\left(
P_e^{\mathrm{total}}(\boldH,\boldz)>a_n
\right)
\leq
\frac{\epsilon_n}{a_n}
\leq
\sqrt{\epsilon_n}
\xrightarrow[n\to\infty]{}0,
\]
or equivalently,
\[
\Pr_{\boldH,\boldz}\left(
P_e^{\mathrm{total}}(\boldH,\boldz)\leq a_n
\right)
\xrightarrow[n\to\infty]{}1.\qedhere
\]
\end{proof}
\subsection{Rate calculation}\label{subsection:rate-calculation}

We now calculate the rate of the construction.  The fixed pilot sequence
carries no message bits, so all message information is contained in the
level codewords.  Fix a choice of the parity-check matrices for which
every \(H_i\) has full row rank.  By
Lemma~\ref{lemma:full rank random parity matrix} and
Definition~\ref{definition:random-linear-erasure-codes}, for every
\(i\in[J]\) we have
$
\dim(\cC_{\mathrm{er}}^i)
=
n_{\mathrm{er}}^i-r_i
=
n_{\mathrm{er}}^iR_{\mathrm{er}}^i
$, and hence
$
|\cC_{\mathrm{er}}^i|
=
2^{n_{\mathrm{er}}^iR_{\mathrm{er}}^i}.
$
The random shift by \(\boldz^{\,i}\) is a bijection, and the
\(\mathrm{RLL}(0,\beta-1)\) encoder is injective.  
Therefore, neither
operation changes the number of level-\(i\) codewords, namely
$
|\bar{\cC}_{\mathrm{er}}^i|
=
|\cC_{\mathrm{sh}}^i|
=
|\cC_{\mathrm{er}}^i|
=
2^{n_{\mathrm{er}}^iR_{\mathrm{er}}^i}.
$
Recall from Definition~\ref{definition:random-linear-erasure-codes} that
\(n_{\mathrm{er}}^i=(1-1/\beta)|\cA_i|\).  It follows that
\[
|\bar{\cC}_{\mathrm{er}}^i|
=
2^{(1-1/\beta)|\cA_i|R_{\mathrm{er}}^i}.
\]

By Lemma~\ref{lemma:periodic-coordinate-properties}, the sets
\(\cA_1,\ldots,\cA_J,\cS_J\) form a disjoint partition of the codeword
positions.  The pilot sequence occupying \(\cS_J\) is fixed, while the
level-\(i\) RLL codeword occupies exactly \(\cA_i\).  Thus, independent
choices of one codeword from each
\(\bar{\cC}_{\mathrm{er}}^i\) produce distinct length-\(n\) codewords,
and consequently
$
|\cC|
=
\prod_{i=1}^{J}|\bar{\cC}_{\mathrm{er}}^i|.
$
Lemma~\ref{lemma:periodic-coordinate-properties} also gives
$
\frac{|\cA_i|}{n}
=
\frac{1}{t_{i-1}}-\frac{1}{t_i}.
$
Therefore, the rate of the resulting code is
\begin{align}
R(\cC)
=
\frac{1}{n}\log |\cC| 
=
\left(1-\frac{1}{\beta}\right)
\sum_{i=1}^{J}
\left(
\frac{1}{t_{i-1}}-\frac{1}{t_i}
\right)
R_{\mathrm{er}}^i.
\label{eq:multilevel-code-rate}
\end{align}

We next lower-bound the weighted sum in
\eqref{eq:multilevel-code-rate} for fixed construction parameters.

\begin{theorem}
\label{theorem:finite-parameter-rate}
The rate of the torn-paper code $\cC$ in Section~\ref{Section:proposed-work encoding} satisfies
\begin{align*}
R(\cC)
\geq
\left(1-\frac{1}{\beta}\right)
\Bigg[
e^{-\alpha(1+\delta)}
-\frac{e^{-\alpha(1+\delta)T}}{T} -
\bigl(1+\alpha(1+\delta)\bigr)
e^{-\alpha(1+\delta)}
\frac{1}{m}
-3\eta\left(1-\frac{1}{T}\right)
\Bigg].
\label{eq:finite-parameter-rate}
\end{align*}

\end{theorem}

\begin{proof}
By Definition~\ref{definition:random-linear-erasure-codes} and
Lemma~\ref{lemma:full rank random parity matrix}, we have
\[
R_{\mathrm{er}}^i
=
\bigl(1+\alpha(1+\delta)t_i\bigr)
e^{-\alpha(1+\delta)t_i}
-3\eta.
\]
Substituting this expression into
\eqref{eq:multilevel-code-rate}, we obtain
\[
\begin{aligned}
\sum_{i=1}^{J}
\left(
\frac{1}{t_{i-1}}-\frac{1}{t_i}
\right)R_{\mathrm{er}}^i
=
\sum_{i=1}^{J}
\left(
\frac{1}{t_{i-1}}-\frac{1}{t_i}
\right)
\bigl(1+\alpha(1+\delta)t_i\bigr)
e^{-\alpha(1+\delta)t_i}
-3\eta
\sum_{i=1}^{J}
\left(
\frac{1}{t_{i-1}}-\frac{1}{t_i}
\right).
\end{aligned}
\]
Since
\[
\sum_{i=1}^{J}
\left(
\frac{1}{t_{i-1}}-\frac{1}{t_i}
\right)
=
\frac{1}{t_0}-\frac{1}{t_J}
\overset{\text{Def.}\ref{definition:alignment-scales}}{=}
1-\frac{1}{T},
\]
it follows that
\begin{equation}\label{equation:expression of rc}
\sum_{i=1}^{J}
\left(
\frac{1}{t_{i-1}}-\frac{1}{t_i}
\right)R_{\mathrm{er}}^i
=
\sum_{i=1}^{J}
\left(
\frac{1}{t_{i-1}}-\frac{1}{t_i}
\right)
\bigl(1+\alpha(1+\delta)t_i\bigr)
e^{-\alpha(1+\delta)t_i}
-3\eta\left(1-\frac{1}{T}\right),
\end{equation}
and it remains to lower-bound
\[
\sum_{i=1}^{J}
\left(
\frac{1}{t_{i-1}}-\frac{1}{t_i}
\right)
\bigl(1+\alpha(1+\delta)t_i\bigr)
e^{-\alpha(1+\delta)t_i}.
\]
Define
$
g_\delta(t)
\triangleq
\bigl(1+\alpha(1+\delta)t\bigr)
e^{-\alpha(1+\delta)t}
$
and
$
h_\delta(t)
\triangleq
\frac{g_\delta(t)}{t^2}.
$
Notice that both functions are decreasing for \(t>0\) and that
$
-\frac{d}{dt}
\left(
\frac{e^{-\alpha(1+\delta)t}}{t}
\right)
=
h_\delta(t)
$.
Recall that \(t_i-t_{i-1}=1/m\) by Definition~\ref{definition:alignment-scales}, and hence
\[
\begin{aligned}
\left(
\frac{1}{t_{i-1}}-\frac{1}{t_i}
\right)g_\delta(t_i)
=
\frac{t_i-t_{i-1}}{t_{i-1}t_i}g_\delta(t_i)=
\frac{1}{m}\frac{g_\delta(t_i)}{t_{i-1}t_i}\overset{(\star)}{\ge}\frac{1}{m}\frac{g_\delta(t_i)}{t_i^2},
\end{aligned}
\]
where~$(\star)$ follow since \(t_{i-1}\leq t_i\).
Moreover, since \(g_\delta(t)\) is decreasing, it follows that
\[
\frac{g_\delta(t_i)}{t_{i-1}t_i}
\leq
\frac{g_\delta(t_{i-1})}{t_{i-1}^2},
\]
and hence
\[
\frac{1}{m}h_\delta(t_i)
\leq
\left(
\frac{1}{t_{i-1}}-\frac{1}{t_i}
\right)
g_\delta(t_i)
\leq
\frac{1}{m}h_\delta(t_{i-1}).
\]
Since \(h_\delta\) is decreasing it follows that $h_\delta(t_i)\leq h_\delta(t)\leq h_\delta(t_{i-1})$ for every \(t\in[t_{i-1},t_i]\), and by integrating over \([t_{i-1},t_i]\) and using \(t_i-t_{i-1}=1/m\), we obtain
\[
\frac{1}{m}h_\delta(t_i)
\leq
\int_{t_{i-1}}^{t_i}h_\delta(t)\,dt
\leq
\frac{1}{m}h_\delta(t_{i-1}).
\]
Therefore, the difference
between the weighted sum and the corresponding integral is at most
\[
\begin{aligned}
\left|
\sum_{i=1}^{J}
\left(
\frac{1}{t_{i-1}}-\frac{1}{t_i}
\right)
g_\delta(t_i)
-
\int_1^T h_\delta(t)\,dt
\right|\leq
\frac{1}{m}
\sum_{i=1}^{J}
\bigl(h_\delta(t_{i-1})-h_\delta(t_i)\bigr)=
\frac{1}{m}
\bigl(h_\delta(1)-h_\delta(T)\bigr)\leq
\frac{1}{m}
h_\delta(1).
\end{aligned}
\]
Since
$
h_\delta(1)
=
\bigl(1+\alpha(1+\delta)\bigr)e^{-\alpha(1+\delta)},
$
and since
\begin{align*}
\int_1^T h_\delta(t)\,dt
=
\left[
-\frac{e^{-\alpha(1+\delta)t}}{t}
\right]_{1}^{T}=
e^{-\alpha(1+\delta)}
-
\frac{e^{-\alpha(1+\delta)T}}{T}.
\end{align*}
it follows that
\begin{equation}\label{eq:integral}
    \sum_{i=1}^{J}
\left(
\frac{1}{t_{i-1}}-\frac{1}{t_i}
\right)
g_\delta(t_i)\geq
e^{-\alpha(1+\delta)}
-
\frac{e^{-\alpha(1+\delta)T}}{T}
-
\bigl(1+\alpha(1+\delta)\bigr)e^{-\alpha(1+\delta)}
\frac{1}{m}.
\end{equation}
Now recall from~\eqref{eq:multilevel-code-rate} and~\eqref{equation:expression of rc} that
\[
\begin{aligned}
R(\cC)
&=
\left(1-\frac{1}{\beta}\right)
\left[
\sum_{i=1}^{J}
\left(
\frac{1}{t_{i-1}}-\frac{1}{t_i}
\right)
g_\delta(t_i)-3\eta\left(1-\frac{1}{T}\right)
\right].
\end{aligned}
\]
Utilizing~\eqref{eq:integral}, we have
\[
\begin{aligned}
R(\cC)
\geq
\left(1-\frac{1}{\beta}\right)
\Bigg[
e^{-\alpha(1+\delta)}
-
\frac{e^{-\alpha(1+\delta)T}}{T}
-
\bigl(1+\alpha(1+\delta)\bigr)e^{-\alpha(1+\delta)}
\frac{1}{m}
-
3\eta\left(1-\frac{1}{T}\right)
\Bigg].
\end{aligned}
\]
Therefore, by first choosing \(\delta>0\) sufficiently small and \(T\)
sufficiently large, and then choosing \(m\) and \(\beta\) sufficiently
large and \(\eta>0\) sufficiently small, the right-hand side can be made
arbitrarily close to \(e^{-\alpha}\).
Hence, for every \(\varepsilon>0\), the parameters can be chosen so that $R(\cC)\geq e^{-\alpha}-\varepsilon$.
\end{proof}

The preceding theorem shows that the rate of our torn-paper code can be made arbitrarily close to the capacity. 
We now choose these parameters and combine the rate bound with the decoding-reliability results proved in the previous subsection.

\begin{theorem}\label{theorem:capacity-achieving-family}
For every \(\varepsilon>0\), the parameters of the construction can be
chosen such that there exists a sequence \(a_n\to0\)
satisfying
\[
\Pr_{\boldH,\boldz}\left(
P_e(\cC_n(\boldH,\boldz))\leq a_n
\ \text{and}\
R(\cC_n(\boldH,\boldz))\geq e^{-\alpha}-\varepsilon
\right)
\xrightarrow[n\to\infty]{}1.
\]
Consequently, there exists a sequence of binary codes
\(\{\cC_n\}_{n\geq1}\) such that
\[
P_e(\cC_n)\xrightarrow[n\to\infty]{}0
\qquad\text{and}\qquad
\liminf_{n\to\infty}R(\cC_n)
\geq e^{-\alpha}-\varepsilon.
\]
Therefore, every rate strictly below \(e^{-\alpha}\) is achievable.
\end{theorem}

\begin{proof}
Fix \(\varepsilon>0\) and choose:
\begin{itemize}
    \item \(\delta>0\) sufficiently small so that
    \(e^{-\alpha(1+\delta)}\geq e^{-\alpha}-\frac{\varepsilon}{8}\);
    \item an integer \(T>1\) sufficiently large so that
    \(\frac{e^{-\alpha(1+\delta)T}}{T}<\frac{\varepsilon}{8}\);
    \item an integer \(m\) sufficiently large so that
    \(\frac{(1+\alpha(1+\delta))e^{-\alpha(1+\delta)}}{m}
    <\frac{\varepsilon}{8}\); and
    \item \(\eta>0\) sufficiently small so that
    \(3\eta\left(1-\frac{1}{T}\right)<\frac{\varepsilon}{8}\)
    and
    $
    3\eta<
    \bigl(1+\alpha(1+\delta)T\bigr)
    e^{-\alpha(1+\delta)T}$; 
    since
    \(R_{\mathrm{er}}^i=
    \bigl(1+\alpha(1+\delta)t_i\bigr)
    e^{-\alpha(1+\delta)t_i}-3\eta\)
    and \(t_i\leq T\), the function
    \(\bigl(1+\alpha(1+\delta)t\bigr)e^{-\alpha(1+\delta)t}\)
    is nonincreasing in \(t\). Therefore,
    \[
    \bigl(1+\alpha(1+\delta)t_i\bigr)
    e^{-\alpha(1+\delta)t_i}
    \geq
    \bigl(1+\alpha(1+\delta)T\bigr)
    e^{-\alpha(1+\delta)T}
    >
    3\eta,
    \]
    which guarantees \(R_{\mathrm{er}}^i>0\) for every \(i\in[J]\).
\end{itemize}

By Lemma~\ref{lemma:full rank random parity matrix}, for every
\(i\in[J]\), the probability that \(H_i\) does not have full row rank
tends to zero as \(n\to\infty\).
Since \(J=m(T-1)\) is fixed after the
above parameters are chosen, a union bound gives
\[
\Pr_{\boldH}\left(
\operatorname{rank}(H_i)=r_i
\text{ for every }i\in[J]
\right)
\to1.
\]

Whenever all \(H_i\)'s have full row rank,
Theorem~\ref{theorem:finite-parameter-rate}, together with the above
choices of \(\delta,T,m,\eta\), gives
\[
R(\cC_n)
\geq
\left(1-\frac{1}{\beta}\right)
\left(e^{-\alpha}-\frac{\varepsilon}{2}\right).
\]
Since \(\beta=\omega(1)\), we have \(1-1/\beta\to1\). Therefore, for
all sufficiently large \(n\),
\[
R(\cC_n)\geq e^{-\alpha}-\varepsilon.
\]
The random shifts do not affect the rate of the code. Hence,
\[
\Pr_{\boldH,\boldz}\left(
R(\cC_n(\boldH,\boldz))
\geq e^{-\alpha}-\varepsilon
\right)
\to1.
\]

Moreover, by Lemma~\ref{lemma:high_probability_existence}, there exists
a deterministic sequence \(a_n\to0\) such that
\[
\Pr_{\boldH,\boldz}\left(
P_e(\cC_n(\boldH,\boldz))\leq a_n
\right)
\to1.
\]
Since both events have probability tending to one, a union bound gives
\begin{equation}\label{eq:event}
    \Pr_{\boldH,\boldz}\left(
P_e(\cC_n(\boldH,\boldz))\leq a_n
\ \text{and}\
R(\cC_n(\boldH,\boldz))
\geq e^{-\alpha}-\varepsilon
\right)
\to1.
\end{equation}

Consequently, for all sufficiently large \(n\), this event~\eqref{eq:event} is nonempty.
For each such \(n\), fix one choice of \((\boldH,\boldz)\) in the event,
and let \(\cC_n\) be the resulting deterministic code. Then
\[
P_e(\cC_n)\leq a_n\to0,
\]
and
\[
\liminf_{n\to\infty}R(\cC_n)
\geq
e^{-\alpha}-\varepsilon.
\]
Hence, for every \(\varepsilon>0\), there exists a deterministic sequence
of codes with vanishing average decoding error probability and asymptotic
rate at least \(e^{-\alpha}-\varepsilon\). Therefore, every rate strictly
below \(e^{-\alpha}\) is achievable.
\end{proof}

\begin{proposition}\label{proposition:polynomial-time-complexity}
Fix \(\varepsilon>0\) and choose the construction parameters as in
Theorem~\ref{theorem:capacity-achieving-family}. For every such code
attaining rate \(e^{-\alpha}-\varepsilon\), the encoder and decoder can
be implemented in \(O(n^3)\) time.
\end{proposition}
\begin{proof}
Once \(\varepsilon\) is fixed, the parameters \(T,m,J\), and \(Q\) are
constants independent of \(n\). Each parity-check matrix \(H_i\) has at
most \(n\) rows and columns. Hence, the parity-check matrices and shift
vectors can be generated in \(O(n^2)\) time. For each full-row-rank
\(H_i\), a generator matrix for \(\ker(H_i)\) can be computed by
Gaussian elimination in \(O(n^3)\) time. Since \(J\) is constant, all
generator matrices can be computed in \(O(n^3)\) time. The De Bruijn
sequence, the pilot sequence, and the sets used by the construction can
all be generated in at most \(O(n^2)\) time. Thus, the total
preprocessing complexity is \(O(n^3)\).

After preprocessing, encoding each level requires multiplication by a
precomputed generator matrix, which takes at most \(O(n^2)\) time.
Since \(J\) is constant, encoding all levels requires \(O(n^2)\) time.
The shifts, RLL encoding, self-interleaving, and placement of the
resulting bits require only \(O(n)\) additional time. Hence, the total
encoding complexity is \(O(n^2)\).

For decoding, the received fragments have total length \(n\), so there
are at most \(n\) fragments. In the initial alignment, a brute-force
implementation tests at most \(n\) candidate positions for each fragment
and compares at most \(n\) bits for each candidate. Therefore, the
initial alignment requires at most \(O(n^3)\) time. At each subsequent
level, Algorithm~\ref{alg:recycled-pilot-alignment} has the same
worst-case complexity \(O(n^3)\). Since \(J\) is constant, all alignment
levels together require \(O(n^3)\) time.

Finally, each level-wise erasure decoder solves a binary linear system of
size at most \(n\times n\), which takes \(O(n^3)\) time by Gaussian
elimination. Since there are only \(J\) levels, all erasure-decoding
steps together require \(O(n^3)\) time. Therefore, the complete encoder
and decoder can be implemented in \(O(n^3)\) time.
\end{proof}
\printbibliography
\appendices
\section{Omitted proofs}\label{appendix:rll-encoders}

\begin{lemma}
There exists an $\mathrm{RLL}(0,k-1)$ fixed-length block code with rate $(k-1)/k$.
\end{lemma}

\begin{proof}
This is the construction of~\cite[Lemma~1]{liu2026improved}. Let $\boldv\in\{0,1\}^{\ell}$, and assume that $k-1$ divides $\ell$; otherwise, pad the final block with at most $k-2$ dummy bits. Partition $\boldv$ into $\ell/(k-1)$ blocks $\boldv^{(1)},\ldots,\boldv^{(\ell/(k-1))}$ of length $k-1$ and append one $1$ to each block: $\boldx^{(j)}\triangleq\boldv^{(j)}1$. The concatenation has length $\ell k/(k-1)$ and rate $(k-1)/k$. Every length-$k$ encoded block ends in $1$, so no zero run has length $k$ or more. Decoding splits the codeword into length-$k$ blocks and removes the last bit of each block.
\end{proof}

\begin{lemma}
There exists a $\widetilde{\mathrm{RLL}}(0,k-2)$ fixed-length block code with rate $(k-2)/k$.
\end{lemma}

\begin{proof}
This is the construction of~\cite[Lemma~2]{liu2026improved}. Let $\boldv\in\{0,1\}^{\ell}$, with $k-2$ dividing $\ell$ after at most $k-3$ padding bits. Partition $\boldv$ into blocks of length $k-2$ and map each block $\boldv^{(j)}$ to $\boldx^{(j)}\triangleq1\boldv^{(j)}1$. The encoded length is $\ell k/(k-2)$ and the rate is $(k-2)/k$. Every block begins and ends in $1$, so every zero run has length at most $k-2$. Decoding splits the codeword into length-$k$ blocks and removes the first and last bit of each block.
\end{proof}
\begin{lemma}\label{lemma:initial-local-alignment-correct-appendix}
For every \(\boldf\in\cF_J\),
Algorithm~\ref{alg:initial-local-alignment} identifies the unique
de-interleaved subsequence that contains bits of the pilot sequence.
Consequently, \(\boldp_{\boldf}\) is exactly the pilot subsequence contained in \(\boldf\), and every \(\boldy_{\boldf,r}\) is the corresponding modulo-\(T\) de-interleaved subsequence of \(\boldf\).
\end{lemma}

\begin{proof}
Let \(s\) be the true starting position of \(\boldf\), and let
\(j^\star\in\bZ_T\) be the unique index satisfying
\(s+j^\star\equiv0\pmod T\). We prove the claim in two steps.

\emph{Step 1: the pilot-containing de-interleaved subsequence contains a
marker.}

By the definition of \(j^\star\), the subsequence \(\boldw_{\boldf,j^\star}\) consists exactly of the fragment bits whose global positions are multiples of \(T\). 
Since \(c[vT]=p[v]\) for every \(v\in\{0,1,\ldots,n/T-1\}\), these bits form a subsequence of the pilot sequence \(\boldp\).
There is exactly one pilot bit in every \(T\) consecutive codeword
bits. 
Therefore, for every \(\boldf\) which has length
\((1+\delta)T\log n\), it contains
\(|\boldw_{\boldf,j^\star}|\geq(1+\delta)T\log n/T
=(1+\delta)\log n\) pilot bits.
Recall that in Definition~\ref{definition:modified-pilot}, the markers in \(\boldp\) exist in every \(b_n=(1+\delta/2)\log n\) pilot bits, and each marker has length \(q_J\beta\). 
Hence every pilot subsequence of length at least \(b_n+q_J\beta-1\) contains a complete marker. 
Since \(q_J\beta=o(\log n)\) as $q_J$ is a constant and $\beta=o(\log n)$, we have \((1+\delta)\log n-1>b_n+q_J\beta-1\) for all sufficiently large \(n\). 
Thus \(\boldw_{\boldf,j^\star}\) contains a complete marker \(0^{q_J\beta}\).
Therefore, at least one de-interleaved subsequence is selected by the marker
test (Line~\ref{line:marker test}) in Algorithm~\ref{alg:initial-local-alignment}.

\emph{Step 2: no information-only de-interleaved subsequence contains a
marker.}

We first show that every \(\boldw_{\boldf,j}\) with \(j\neq j^\star\)
contains only bits from the level codewords.  
Let \(j\neq j^\star\) and set \(d\triangleq(s+j)\bmod T\).  
Then \(d\in\{1,\ldots,T-1\}\), and every bit in \(\boldw_{\boldf,j}\) is located at a global position congruent to \(d\) modulo \(T\).  
Since \(Q=q_JT\), within each period of length \(Q\) these positions have residues $d,\ d+T,\ d+2T,\ldots,d+(q_J-1)T \pmod Q.$
By Definition~\ref{definition:nested-residue-sets}, the pilot residues are
$\cR_J=\{0,T,2T,\ldots,(q_J-1)T\}.$
Since \(d\in\{1,\ldots,T-1\}\), none of the above \(q_J\) residues belongs
to \(\cR_J\).  Hence every bit in \(\boldw_{\boldf,j}\) comes from one of
the level codewords, and \(\boldw_{\boldf,j}\) contains no bits from the
fixed pilot sequence.

We next show that \(\boldw_{\boldf,j}\) cannot contain \(q_J\beta\)
consecutive zeros. 
As the bits of \(\boldw_{\boldf,j}\) are taken every
\(T\) codeword bits, their residues modulo \(Q\) cycle through
$d,\ d+T,\ldots,d+(q_J-1)T.$
Therefore, every \(q_J\beta\) consecutive bits of
\(\boldw_{\boldf,j}\) contain exactly \(\beta\) bits from each of these
\(q_J\) residue classes. 
Fix one such residue \(r\) and since \(r\notin\cR_J\), there is a unique \(i\in[J]\) such that \(r\in\cB_i\).  
The \(\beta\) bits corresponding to \(r\) occur consecutively in the subsequence $(c[r],c[Q+r],c[2Q+r],\ldots).$
By the self-interleaving construction of the level-\(i\) codeword, this
 subsequence is a subsequence of \(\bar{\bolds}^{\,i}\).  
 Since \(\bar{\bolds}^{\,i}\) satisfies the \(\mathrm{RLL}(0,\beta-1)\) constraint, no \(\beta\) consecutive bits in this subsequence can all be zero.
 Hence no \(\boldw_{\boldf,j}\) with \(j\neq j^\star\) can contain
\(0^{q_J\beta}\).

Combining the two steps, \(\boldw_{\boldf,j^\star}\) contains
\(0^{q_J\beta}\), whereas every other de-interleaved subsequence does not.
Algorithm~\ref{alg:initial-local-alignment} therefore identifies
\(j_{\boldf}=j^\star\) uniquely and returns
\(\boldp_{\boldf}=\boldw_{\boldf,j^\star}\), which is exactly the pilot
subsequence contained in \(\boldf\).
Finally, for every \(i\in\bZ_T\), the sequence
\(\boldy_{\boldf,i}
=\boldw_{\boldf,(j^\star+i)\bmod T}\) contains exactly the fragment bits
whose global positions are congruent to \(i\) modulo \(T\). 
Thus the returned pilot subsequence and all returned modulo-\(T\) de-interleaved subsequences are correct.
\end{proof}
\section{Supporting estimates for rate and reliability}\label{appendix:rate-reliability-estimates}

This section contains lemmas necessary for Section~\ref{Section:rate-and-reliability}.
\begin{lemma}\label{lemma:number-of-fragments}
 Let $K$ be the smallest index such that
\[
    \sum_{i=1}^{K} N_i \ge n,
\]
where $N_1,N_2,\ldots$ are i.i.d. $\operatorname{Geometric}(p_n)$ random variables representing the lengths of broken fragments. Then
\[
    K = O\!\left(\frac{n}{\log n}\right)
\]
with high probability.
\end{lemma}

\begin{proof}
For each \(t\in\{1,\dots,n-1\}\), let \(B_t\) be the indicator of the
event that a break occurs between \(x_t\) and \(x_{t+1}\). Then
\(B_1,\dots,B_{n-1}\) are independent
\(\operatorname{Bernoulli}(p_n)\) random variables. Since every break
increases the number of fragments by one,
\[
K=1+\sum_{t=1}^{n-1}B_t.
\]

Let
\[
Y_n\triangleq\sum_{t=1}^{n-1}B_t.
\]
Then
\[
Y_n\sim\operatorname{Binomial}(n-1,p_n)
\qquad\text{and}\qquad
\mathbb E[Y_n]=(n-1)p_n.
\]
Since \(\lim_{n\to\infty}p_n\log n=\alpha\), there exists a constant
\(c>0\) such that
\[
p_n\leq\frac{c}{\log n}
\]
for all sufficiently large \(n\). Hence,
\[
\mathbb E[Y_n]\leq\frac{c(n-1)}{\log n}.
\]

Therefore,
\[
\frac{2c(n-1)}{\log n}\geq2\mathbb E[Y_n].
\]
If \(\mathbb E[Y_n]>0\), write
\[
\frac{2c(n-1)}{\log n}
=
(1+\gamma_n)\mathbb E[Y_n],
\]
where \(\gamma_n\geq1\). By the Chernoff bound,
\[
\Pr\left(
Y_n>\frac{2c(n-1)}{\log n}
\right)
\leq
\exp\left(
-\frac{\gamma_n\mathbb E[Y_n]}{3}
\right).
\]
Moreover,
\[
\gamma_n\mathbb E[Y_n]
=
\frac{2c(n-1)}{\log n}-\mathbb E[Y_n]
\geq
\frac{c(n-1)}{\log n},
\]
and therefore
\[
\Pr\left(
Y_n>\frac{2c(n-1)}{\log n}
\right)
\leq
\exp\left(
-\frac{c(n-1)}{3\log n}
\right)
\xrightarrow[n\to\infty]{}0.
\]
If \(\mathbb E[Y_n]=0\), then \(Y_n=0\) almost surely, so the conclusion
is immediate.

Thus, with high probability,
\[
K
=
1+Y_n
\leq
1+\frac{2c(n-1)}{\log n}
=
O\!\left(\frac{n}{\log n}\right).
\qedhere
\]
\end{proof}

\begin{lemma}\label{lemma:long-fragment-coverage}
For a constant~$\gamma$, the coverage~$V_\gamma$ satisfies the following two properties
    \begin{align}\label{Equation: limit of coverage}
        \Pr\!\left( \left| V_{\gamma} - (\alpha\gamma + 1)e^{-\alpha\gamma} \right| > \epsilon \right)
\;\xrightarrow[n\to\infty]{} 0,\qquad \text{and} \qquad \lim_{n \to \infty} \mathbb{E}[V_\gamma]
    = (\alpha \gamma + 1)e^{-\alpha \gamma}.
\end{align}
\end{lemma}

\begin{proof}
This is the long-fragment coverage result used in
\cite[Lemma~15]{liu2026improved}, which follows from
\cite[Lemma~6]{shomorony2021torn}.
\end{proof}

\begin{lemma}\label{lemma:number-of-debruijn-subsequences-in-fragment}
Fix $i\in[J-1]$ and consider a fragment $\boldf$ of length
$(1+\delta)t_i\log n$.
Consider all candidate length-$k_p$ De Bruijn subsequences that can be
extracted from $\boldf$ under all possible candidate placements.
The number of such candidate subsequences, counted with multiplicity
according to how they arise from the fragment, is at most
\(Tb_n(1+\delta)t_i\log n=n^{o(1)}\).
\end{lemma}

\begin{proof}
We first explain what is being counted.
As the candidate starting position \(a\) varies, different bits of the
fragment may be selected as bits of the pilot sequence, and among these
pilot bits, different bits may be selected as De Bruijn bits after the
marker bits and the \(1\)'s added by the
\(\widetilde{\mathrm{RLL}}(0,\beta-2)\) encoder are removed.
We count candidate length-\(k_p\) De Bruijn subsequences according to
how they arise from the fragment. In particular, two candidate
subsequences are counted separately if they arise from different
candidate pilot subsequences, different positions within a length-\(b_n\)
pilot block, or different choices of the length-\(k_p\) De Bruijn
portion, even if the resulting binary sequences are identical.

A candidate starting position \(a\) first determines which fragment bits
are mapped to the pilot sequence.
Since one pilot bit is placed every
\(T\) bits, there are \(T\) possibilities, corresponding
to the \(T\) residue classes of the fragment modulo \(T\).
After a candidate pilot subsequence has been selected, we must determine
which of its bits are from the De Bruijn sequence rather than marker bits
or the \(1\)'s added by the
\(\widetilde{\mathrm{RLL}}(0,\beta-2)\) encoder.
Recall that the pilot
sequence is constructed in blocks of \(b_n\) pilot bits
(as in Definition~\ref{definition:modified-pilot}), and every such
block has exactly the same structure: the marker occupies fixed positions
within the block, and the \(1\)'s added by the
\(\widetilde{\mathrm{RLL}}(0,\beta-2)\) encoder also occur at fixed
positions.
Therefore, it is enough to know which one of the
\(b_n\) positions \(0,1,\ldots,b_n-1\) within a pilot block corresponds
to the first bit of the candidate pilot subsequence.
Once this position is fixed, the construction uniquely determines,
throughout the entire candidate pilot subsequence, which bits are marker
bits, which bits are RLL-added \(1\)'s, and which bits are original
De Bruijn bits.

Now fix one of these \(b_n\) possible positions.
For each of the \(T\) candidate pilot subsequences, remove all marker
bits and all \(1\)'s added by the
\(\widetilde{\mathrm{RLL}}(0,\beta-2)\) encoder.
The remaining bits are from the De Bruijn sequence \(\boldq\).
For each such case, there are at most
\((1+\delta)t_i\log n\) possible choices of a length-\(k_p\)
De Bruijn subsequence within the fragment.

Finally, there are \(T\) possible pilot subsequences, \(b_n\) possible
positions for the first bit of a candidate pilot subsequence within a
length-\(b_n\) pilot block, and for each such case there are at most
\((1+\delta)t_i\log n\) possible length-\(k_p\) De Bruijn subsequences.
Hence, the total number of candidate length-\(k_p\) De Bruijn
subsequences, counted with the multiplicity described above, is at most
$
Tb_n(1+\delta)t_i\log n.
$
Since \(b_n=(1+\delta/2)\log n\) and \(T,t_i\) are constants, this upper
bound is
\[
T\left(1+\frac{\delta}{2}\right)(1+\delta)t_i(\log n)^2
=
O((\log n)^2)
=
n^{o(1)}.\qedhere
\]
\end{proof}
\end{document}